\documentclass[a4paper,11pt]{article}
\usepackage[margin=1in]{geometry}

\usepackage{amssymb}
\usepackage{mathtools}
\usepackage{amsthm}
\usepackage{tikz}
\usepackage{subcaption}
\usepackage{enumitem}
\usepackage[ruled]{algorithm2e}
\usepackage{xcolor}
\usepackage{microtype}
\usepackage{setspace}

\newlength{\algaofontsize}
\usepackage[
    numbers,
    square,
]{natbib}

\usepackage[
    pdfusetitle,
    colorlinks,
    linkcolor=black,
    citecolor=black,
    urlcolor=black,
]{hyperref}

\usepackage[
    capitalise,
]{cleveref}

\usetikzlibrary{
    math,
}

\newtheorem{theorem}{Theorem}[section]

\newtheorem{lemma}[theorem]{Lemma}

\theoremstyle{definition}

\newcommand{\defas}{\ensuremath{\coloneqq}}

\newcommand{\N}{\mathbb{N}}

\newcommand{\calA}{\mathcal{A}}

\newcommand{\calS}{\mathcal{S}}

\newcommand{\calX}{\mathcal{X}}

\newcommand{\bfA}{\mathbf{A}}

\newcommand{\bfX}{\mathbf{X}}
\newcommand{\bfv}{\mathbf{v}}
\newcommand{\bfx}{\mathbf{x}}
\newcommand{\bfz}{\mathbf{z}}

\newcommand{\alg}{\textsc{Alg}}
\newcommand{\opt}{\textsc{Opt}}
\newcommand{\algs}{\textsc{Select}}

\newcommand{\alga}{\textsc{Assign}}
\newcommand{\gen}{\textsc{Gen}}
\newcommand{\gena}{\textsc{Gen}^{\textnormal{AS}}}
\newcommand{\score}{\sigma}
\newcommand{\selected}{Z}
\newcommand{\ie}{i.e.,\xspace}

\newcommand{\indx}{\ensuremath{t}}
\newcommand{\disjoint}{\ensuremath{\mathbin{\dot\cup}}}
\newcommand{\range}[1]{\ensuremath{{[#1]}_0}}

\DeclareMathOperator*{\argmax}{arg\,max}
\DeclareMathOperator*{\argmin}{arg\,min}

\definecolor{darkred}{rgb}{0.545,0,0}
\definecolor{myorange}{rgb}{1,0.549,0}
\definecolor{navy}{rgb}{0,0,0.502}

\newcommand{\cola}{\textcolor{gray!75}{0}}
\newcommand{\colb}{\textcolor{navy}{1}}
\newcommand{\colc}{\textcolor{myorange}{2}}
\newcommand{\cold}{\textcolor{darkred}{3}}

\newcommand{\Description}[1]{}

\begin{document}

\title{Deterministic Impartial Selection with Weights}

\author{
   Javier Cembrano\textsuperscript{1}
   \and
   Svenja M.\ Griesbach\textsuperscript{1}
   \and
   Maximilian J.\ Stahlberg\textsuperscript{1,2}
}
\date{
    \medskip
        \small \textsuperscript{1}Technische Universit\"{a}t Berlin, Germany \\
        \small \texttt{\(\{\)cembrano,\;griesbach\(\}\)@math.tu-berlin.de} \\
    \medskip
        \small \textsuperscript{2}Technische Universit\"{a}t Hamburg, Germany \\
        \small \texttt{maximilian.stahlberg@tuhh.de} \\
}

\hypersetup{
    pdfauthor = {Javier Cembrano, Svenja M. Griesbach, Maximilian J. Stahlberg},
    pdfsubject = {Selecting a subset of agents who nominate each other deterministically and in a strategy-proof manner},
    pdfkeywords = {impartial selection, mechanism design, approximation algorithms},
}

\maketitle

\begin{abstract}
    In the impartial selection problem, a subset of agents up to a fixed size $k$ among a group of $n$ is to be chosen based on votes cast by the agents themselves.
    A selection mechanism is \emph{impartial} if no agent can influence its own chance of being selected by changing its vote.
    It is \emph{$\alpha$-optimal} if, for every instance, the ratio between the votes received by the selected subset is at least a fraction of $\alpha$ of the votes received by the subset of size $k$ with the highest number of votes.
    We study deterministic impartial mechanisms in a more general setting with arbitrarily weighted votes and provide the first approximation guarantee, roughly $1/\lceil 2n/k\rceil$.
    When the number of agents to select is large enough compared to the total number of agents, this yields an improvement on the previously best known approximation ratio of $1/k$ for the unweighted setting.
    We further show that our mechanism can be adapted to the impartial assignment problem, in which multiple sets of up to $k$ agents are to be selected, with a loss in the approximation ratio of $1/2$.
\end{abstract}

\section{Introduction}

Votes and referrals are a key mechanism in the self-organization of communities: political parties elect their representatives, researchers review and rate each other's manuscripts, and hyperlinks on the web attribute topical relevance to an external resource.
Oftentimes, the agents who give the recommendations are themselves interested in being within a top-rated fraction of their group: to occupy a prestigious position, be invited to a conference, or to have a website appear more prominently in search results.
Objectives like these provide an incentive to deviate from a fair evaluation of one's peers.
In particular, agents might omit a recommendation for an immediate contender in order to be ranked above them when the votes are counted.

In a seminal work, \citet{alon2011sum} initiated the search for impartial mechanisms to aggregate the votes cast by \(n\)~agents who want to elect \(k\)~individuals among them, which we refer to as the exact \((n, k)\)-selection problem.
The authors require that no agent is able to influence their own chance of being selected by adjusting the subset of peers that they vote for, while, at the same time, the agents selected by the mechanism should receive an expected sum of votes that is close to that of the highest voted subset of size~\(k\).
We refer to the first condition as \emph{impartiality} and to the second as \emph{\(\alpha\)-optimality}, where \(\alpha \in [0, 1]\) denotes the performance guarantee.
If the mechanism is allowed to make use of randomness and agents may vote for any subset of their peers, then the best known performance guarantee is \[\frac{k}{k+1}\left(1-{\left(\frac{k-1}{k}\right)}^{k+1}\right),\] which gives~\(1/2\) for the selection of a single agent and approaches~\(1 - 1/e\) as \(k \to \infty\)~\citep{bjelde2017impartial}.
It is also known that no impartial mechanism can be better than \(k / (k + 1)\)-optimal, which is tight only for \(k = 1\).
We discuss variants with a limited number of votes per participant as related work.

The problem only becomes more difficult in the deterministic setting, where the mechanism is forced to choose one agent over another even for highly symmetric input.
The instance in which two agents vote for each other and one of them shall be selected requires the mechanism to break the tie, based on an external preference list, in favor of one of the agents.
Impartiality demands that the same agent must be selected also when the other agent withdraws its vote.
But then, an agent with no votes is selected, even though the other agent still receives one.
This yields a performance guarantee of zero for the selection of a single agent in the worst case.
Even for \(k > 1\), no positive performance guarantee is possible~\citep{alon2011sum}, unless, surprisingly, when the mechanism is allowed to select less than \(k\) agents in some instances.
In this case an algorithm achieving \(\alpha = 1/k\) is known~\citep{bjelde2017impartial}.
We refer to this relaxation as the \emph{inexact} \((n, k)\)-selection problem.
Since this insight, the gap towards the best known upper bound, which is \((k-1)/k\) in the inexact selection setting, remained remarkably wide.

More generally, the selection problem allows for votes to be weighted: one then compares the total weight of the selected agents to that of the maximum-weight subset of size $k$.
In a peer review setting, reviewers are often asked to rate the manuscript under consideration on a point scale that ranges from a recommendation to reject to a claim of excellence.
An editor or program chair would then aggregate these scores and accept a limited number of highly rated submissions.
While the established rule to disclose any conflicts of interest protects, if obeyed, against abuse based on personal ties, authors whose papers are on the verge of selection might still profit from giving ratings below their honest estimate, unless the selection mechanism is impartial.
In this setting, although computational studies have been made~\citep{aziz2019strategyproof}, no deterministic mechanism providing a worst-case guarantee was known to date.

\subsection{Our Contribution}

We propose a deterministic impartial mechanism for the inexact $(n,k)$-selection problem that can be applied in the weighted setting.
The agents selected by the mechanism receive at least a fraction of around $k/(2n)$ of the number of votes of the $k$ top-voted agents, whenever $k\geq 2\sqrt{n}$.
More precisely, the performance guarantee is \(\alpha = 1 / \lceil 2n/k \rceil\) for \(k \geq 2\sqrt{n}\) even and \(\alpha = (k - 1)/(k \lceil 2n / (k - 1) \rceil)\) for \(k \geq 2\sqrt{n} + 1\) odd.
For example, the mechanism asymptotically guarantees \(\alpha = 1/4\) when selecting at most half and \(\alpha = 1/3\) when selecting at most two thirds of the agents.
These are the first lower bounds for deterministic selection with weights.
In its applicable range, the mechanism further improves upon the previous best bound of \(1/k\) in the unweighted setting.
Notably, it selects at least $\lfloor k/2\rfloor$ candidates while the previous best mechanism selects either one or two agents, depending on the instance.
The improvement is most noticeable when $k$ is large, where the gap between the previously best known lower and upper bounds of \(1/k\) and \((k-1)/k\), respectively, has been widest.
The construction is best behaved whenever \(b \defas 2n/k \in \N\) and \(b \leq k/2 \in \N\): here a guarantee of \(\alpha = 1/b\) is provided and the analysis of the mechanism is tight.
The mechanism uses a well-structured set of partitions of the agents, whose existence we study in \Cref{sec:partitions} using a connection to hypergraph theory and graph coloring.
The mechanism itself and the proof of the approximation guarantee are presented in \Cref{sec:selection}.

In \Cref{sec:assignment}, we show how the mechanism can be adapted to assign agents to multiple size-bounded subsets.
For example, the participants of a scientific conference could be tasked to assign contributed presentations to thematic sessions or streams.
Some of these slots might be regarded as particularly prestigious, while others may be scheduled near the end of the conference, where there is little time left for discussion over coffee.
To ensure that recommendations focus on a good content-wise match, an impartial mechanism may be employed.
More generally, the size-bounded subsets might represent tasks to distribute, committees to form, or administrative roles to assign.
In this setting, our adjusted mechanism loses only a factor of~\(1/2\) in the performance guarantee, independent of the number of subsets to populate.

\subsection{Related Work}

Impartiality as a desirable axiom in multi-agent problems was introduced by \citet{de2008impartial} and first studied in the context of peer selection by both \citet{holzman2013impartial} and \citet{alon2011sum}: The work by \citeauthor{holzman2013impartial} studied the existence of impartial mechanisms satisfying further axioms such as unanimity and notions of monotonicity, while the research by \citeauthor{alon2011sum} showed that no deterministic impartial mechanism aiming to select exactly $k$ agents can achieve any constant approximation ratio.
Both works consider mechanisms that rely on partitioning the agents, which is also the basis of our mechanism.
While \citet{alon2011sum} use partitions only in the context of randomized mechanisms, \citeauthor{holzman2013impartial} employ them also for deterministic selection, although with different axioms than approximate optimality in mind.
In response to the impossibility result of \citeauthor{alon2011sum}, \citet{bjelde2017impartial} showed that when fewer than \(k\) agents may be selected, $1/k$-optimality is guaranteed by the \emph{bidirectional permutation} mechanism.
The authors further proved an upper bound of $(k-1)/k$ for any deterministic impartial mechanism.

Continuing the axiomatic line, \citet{tamura2014impartial} studied the $k$-selection problem in the single-nomination setting and showed that impartiality is compatible with two natural notions of unanimity.
Their mechanism was extended to the case of a higher, but constant, maximum number of nominations by \citet{cembrano2022optimal}.
Further, \citet{aziz2019strategyproof} proposed a mechanism satisfying certain monotonicity properties and confirmed its performance in a computational study.

Several works have focused on randomized impartial selection.
\citeauthor{alon2011sum} proposed a family of mechanisms based on a random partition of the agents that yield the first lower bounds on the approximation ratio for this setting, namely $1/4$ for $k=1$ and $1-O(1/\sqrt[3]{k})$ for general $k$.
They also provided respective upper bounds of $1/2$ and $1-\Omega(1/k^2)$.
\citet{fischer2015optimal} closed the gap for $k=1$ by giving a $1/2$-approximation algorithm.
\citet{bousquet2014near} designed a mechanism with an approximation guarantee that goes to one as the maximum score of an agent goes to infinity.
A restricted variant of particular importance, first studied in the work of~\citeauthor{holzman2013impartial}, arises when each agent can vote for exactly one other agent.
Here, \citeauthor{fischer2015optimal} provided both lower and upper bounds which were later improved by \citet{cembrano2023improved}.

A setting closely related to the impartial selection of $k$~agents is that of \emph{peer review} in which, in contrast to the classic $k$-selection problem, the votes are weighted and represent a score assigned to a submission.
\citet{kurokawa2015impartial} studied a model where first a limited number of weighted votes is sampled and then the selection is performed.
The authors proposed an impartial randomized mechanism providing a constant approximation ratio with respect to the (non-impartial) mechanism that randomly samples the votes and selects the best possible set of $k$ agents given these votes.\@ \citet{mattei2021peernomination} studied this problem from an axiomatic and experimental point of view, while \citet{lev2023peer} extended this work to the setting with noisy assessments.\@ \citet{dhull2022strategyproofing} explored the scope and limitations of partition-based mechanisms for peer review in terms of approximating the selection of the best $k$ papers.

Beyond multiplicative approximation, some works have studied the scope and limitations of impartial mechanisms in terms of additive guarantees \citep{caragiannis2022impartial, caragiannis2023prior, cembrano2022impartial} and additional economic axioms \citep{edelman2021new, MacK20a}.
Impartiality has also been considered for the selection of agents where preferences come from correlated types \citep{niemeyer2022simple}, for the selection of vertices in graphs with maximal progeny \citep{babichenko2020incentive, zhang2021incentive, zhao2023incentive}, and for generating social rankings of agents who rank each other \citep{kahng2018ranking}.
For a survey on incentive handling in peer mechanisms, see \citet{olckers2022manipulation}.

\section{Preliminaries}\label{sec:prelims}

For $n \in \N \defas \mathbb{Z}_{\geq 1}$, we define the ranges $[n] \defas \{1,\ldots,n\}$ and $\range{n} \defas \{0, \ldots, n - 1\}$, and we write \(\calA_n\) for the set of non-negative $n\times n$ matrices with zero diagonal.
An instance of the weighted selection problem is fully described by an integer $k$ and a weight matrix $A\in \calA_n$, where $k$ is the number of agents to be selected and $A_{ij}$ corresponds to the weight of the vote that agent $i$ casts for agent $j$.
For $A\in \calA_n$, we write $A_{-i}$ for the matrix obtained when removing the \(i\)-th row of \(A\).
Given \(A \in \calA_n\) and \(R,S \subseteq [n]\), we write
\[
    \score_R(S;A) \defas \sum_{\substack{i \in R,\; j \in S}} A_{ij}
\]
for the score of the agents in \(S\) limited to \(R\), and \(\score(S;A)\) short for \(\score_{[n]}(S;A)\).
We omit the weight matrix $A$ whenever it is clear from the context, and we write \(j\) short for \(S = \{j\}\) in the above definitions.
Let $n, k \in \N$ with $k < n$ in the following.
For $A \in \calA_n$, we let
\[
    \opt_k(A) \defas \argmax_{S \subseteq [n] \colon |S| = k} \score(S;A)
\]
denote an arbitrary set with the largest score among subsets of agents of size~\(k\).
We write just \(\opt_k\) when the weight matrix is clear.

An (inexact) $(n,k)$-selection mechanism is a function $f \colon \calA_n \to 2^{[n]}$ such that $|f(A)|\leq k$ for every $A\in \calA_n$.
Such a mechanism is \emph{impartial} if, for every pair of instances $A, A' \in \calA_n$ and for all $i\in [n]$ such that $A_{-i} = A'_{-i}$, it holds that $f(A) \cap \{i\} = f(A')\cap \{i\}$.
We further call an $(n,k)$-selection mechanism \emph{$\alpha$-optimal} if
\[
    \frac{\score(f(A);A)}{\score(\opt_k(A);A)} \geq \alpha
\]
holds for all $A \in \calA_n$ and some $\alpha \in [0,1]$.

We write \(E \disjoint F\) for the disjoint union of sets \(E\) and \(F\).
For a multiset~\(E\), we write \(\mu_E(e)\)
for the multiplicity of \(e \in E\) and \(\mu(E)\)~for the cardinality of~\(E\).

A hypergraph is a pair \(H = (V, E)\) where \(V\)~is a finite set of \emph{vertices} and where \(E \subseteq 2^V\) is a multiset of \emph{(hyper-)edges}. 
We say that \(H\) is \emph{\(d\)-regular} if each vertex is contained in exactly $d$ edges, \ie{} \(\mu\left(\{e \in E \mid v \in e\}\right) = d\) for all \(v \in V\);
\emph{\(b\)-uniform} if each edge contains exactly $b$ vertices, \ie{} \(|e| = b\) for all \(e \in E\); and \emph{linear} if two distinct edges intersect in at most one vertex, \ie{} \(|e_1 \cap e_2| \leq 1\) for all \(e_1, e_2 \in E\) with \(\mu_E(e_1) > 1\) or \(e_1 \neq e_2\).
The \emph{dual} of \(H\) is \(H^* = (E, X)\) where \(X \defas \{\{e \in E \mid v \in e\} \mid v \in V\}\) is a multiset of sets.
One may think of the dual graph in terms of the vertex--edge incidence matrix, which is transposed when taking the dual graph. 
Note that the dual graph may have repeated edges and loops even if the original graph does not have either.

We call a \(2\)-uniform hypergraph without repeated edges a (simple) graph.
For a graph \(G = (V, E)\), an edge \(b\)-coloring is a mapping \(\pi \colon E \to [b]\).
It is \emph{feasible} if \(\pi(e_1) \neq \pi(e_2)\) for all \(e_1, e_2 \in E\) with \(e_1 \cap e_2 \neq \emptyset\).
Likewise, a vertex \(b\)-coloring is a mapping \(\pi \colon V \to [b]\) that we call feasible if \(\pi(u) \neq \pi(v)\) for all \(u, v \in V\) such that \(u, v \in e\) for some \(e \in E\).

\section{Partition Systems}\label{sec:partitions}

The present work takes inspiration from the \emph{partition mechanism}.
This mech\-an\-ism was first proposed by \citet{alon2011sum} for the setting of randomized $(n,1)$-selection, and variants for selecting more than one agent have been studied by \citet{bjelde2017impartial}, \citet{aziz2019strategyproof}, and \citet{xu2019strategyproof}.
In its original formulation due to \citeauthor{alon2011sum}, the partition mechanism assigns each agent into a \emph{voter set}~$S_1$ and a \emph{candidate set}~$S_2$ uniformly at random.
It then considers only votes from agents in $S_1$ to agents in $S_2$ and selects an agent from $S_2$ with maximum score.
This mechanism is impartial as it considers only votes of agents with no chance of being selected, and it is $1/4$-optimal, intuitively, as we see every fourth vote in expectation.
The $(n,k)$-selection variant by \citet{bjelde2017impartial} partitions the agents into $k$ sets instead of two and selects one agent from each set that has the highest score from all other sets, additionally considering internal votes that are directed from left to right according to a random permutation of the agents.
This variant preserves impartiality and provides a guarantee that varies from $1/2$ to $1-1/e$ as $k$ grows from $1$ to infinity.

The partition mechanism, although achieving a good ratio when randomization is possible, performs poorly in the deterministic setting.
If agents are assigned in any fixed way, votes may be adversarially placed between agents in the same set (and opposite to the permutation of the agents if such a step is considered), so that the mechanism cannot do any better, in the worst case, than selecting agents with no votes, while the maximum score may be arbitrarily high.

In the following, we build the foundation for a partition-based $(n,k)$-selection mechanism that is robust against such adversarial placement of votes.
To achieve this, agents appear in the candidate set of more than one partition and with a disjoint set of contenders each time.
This way, votes not seen for a candidate agent in one partition will be seen in another partition wherein that agent re-appears as a candidate.
Of course, repeated candidacy may lead to the same agent being selected multiple times, at the expense of contenders with a high number of votes.
To minimize this possibility, we let every agent contest just twice and we remove duplicate votes.
As our goal is to select up to \(k\)~agents, we define \(k\)~such partitions.
For now, we make also the simplifying assumption that \(n\)~and~\(k\) allow the candidate sets to have equal size~\(b\).
This is without loss of generality as we may fill smaller partitions with dummy agents who cast and receive no votes and are disfavored when breaking ties.
We call a collection of partitions meeting these requirements a \emph{balanced partition system}.

A partition into voters and candidates is fully described by either set.
A balanced partition system may thus be written as a family~\(E\) of candidate subsets of the set of agents~\(V\) or, in other words, as a hypergraph \(H = (V, E)\) without repeated edges, where each $e \in E$ is the candidate set of a single partition.
To fulfill the requirements of a balanced partition system, \(H\)~has to be \(2\)-regular, so that every agent appears in exactly two candidate sets, and \(b\)-uniform, so that all candidate sets $e \in E$ have the same size \(|e| = b\).
The remaining requirement that no two agents compete twice against each other, formally \( |e_1 \cap e_2| \leq 1\) for all $e_1,e_2\in E$ with \( e_1 \neq e_2\),
translates to \(H\)~being linear.
The following lemma implies that we can represent a partition system further by a simple graph.

\begin{lemma}\label{lem:simple_dual}
    A hypergraph is \(2\)-regular and linear if and only if its dual is a simple graph.
\end{lemma}

\begin{proof}
    Let \(H = (V, E)\) be a \(2\)-regular and linear hypergraph.
    Its dual graph is \(H^* = (E, X)\) where \(X = \{\{e \in E \mid v \in e\} \mid v \in V\}\) is a multiset of sets.
    We show that \(H^*\) is a graph, \ie{} that \(H^*\) is \(2\)-uniform and has no repeated edges.
    Let \(x \in X\).
    Then, \(x = \{e \in E \mid v \in e\}\) for some \(v \in V\).
    Since \(H\) is \(2\)-regular, we have \(|x| = 2\) for all \( x\in X\), so \(H^*\) is \(2\)-uniform.
    It remains to show that \(H^*\) has no repeated edges.
    Assume towards a contradiction that there is an \(x \in X\) with multiplicity at least two.
    Then, there are \(v_1 \neq v_2 \in V\) with \(x = \{e \in E \mid v_1 \in e\} = \{e \in E \mid v_2 \in e\}\).
    Since \(H\) is \(2\)-regular, again \(|x| = 2\) holds.
    Let thus \(e_1 \neq e_2 \in E\) with \(v_1, v_2 \in e_1\) and \(v_1, v_2 \in e_2\).
    Then, \(\{v_1, v_2\} \subseteq e_1 \cap e_2\), hence \(|e_1 \cap e_2| \geq 2\) contradicts that \(H\) is linear.

    Let next \(G = (V, E)\) be a graph.
    Its dual graph is a hypergraph \(G^* = (E, X)\) with \(X\) defined as before.
    We show that \(G^*\) is \(2\)-regular and linear.
    To this end, let \(e \in E\) be a vertex of \(G^*\).
    Since \(G\) is \(2\)-uniform, it is \(e = \{v_1, v_2\}\) with \(v_1 \neq v_2\) and thus
    \begin{align*}
        \deg_{G^*}(e)
        &= \mu\left(\{x \in X \mid e \in x\}\right) \\
        &= \mu\left(\{x \in \{\{f \in E \mid v \in f\} \mid v \in V\} \mid e \in x\}\right) \\
        &= \mu\left(\{\{f \in E \mid v \in f\} \mid v \in V \land e \in \{f \in E \mid v \in f\}\}\right) \\
        &= \mu\left(\{\{f \in E \mid v \in f\} \mid v \in V \land v \in e\}\right) \\
        &= \mu\left(\{\{f \in E \mid v_1 \in f\},  \{f \in E \mid v_2 \in f\}\}\right) = 2,
    \end{align*}
    so \(G^*\) is \(2\)-regular.
    Finally, assume towards a contradiction that \(G^*\) is not linear.
    Then, there are \(v_1 \neq v_2 \in V\) such that \(x_1 \coloneqq \{e \in E \mid v_1 \in e\} \in X\) and \(x_2 \coloneqq \{e \in E \mid v_2 \in e\} \in X\), possibly \(x_1 = x_2\) as \(X\) is a multiset, and \(|x_1 \cap x_2| \geq 2\).
    Since \(G\) is simple, \(E\) is a set, so there are \(e_1 \neq e_2 \in E\) with \(v_1, v_2 \in e_1\) and \(v_1, v_2 \in e_2\).
    Since \(G\) is \(2\)-uniform, \(e_1 = \{v_1, v_2\} = e_2\), a contradiction to \(G\) being simple.
\end{proof}

By \Cref{lem:simple_dual} and the fact that order and size as well as degree and rank are dual for hypergraphs, there is a one-to-one correspondence between balanced partition systems where \(n\)~agents are distributed among \(k\)~candidate sets of size~\(b\) on the one hand, and \(b\)-regular simple graphs of order~\(k\) and size~\(n\) on the other hand.
In the simple graph representation, edges correspond to agents while incident vertices correspond to candidate sets that the agents appear in.

In the analysis of the mechanism, we will bound the weight selected by it by that of a subset~\(U\) of top-voted agents that pairwise do not compete.
More precisely, \(U\)~will be a set of maximum weight among a partition of the \(k\)~top-voted agents into \(b\)~many subsets with this property.
If the mechanism does not select some agent~\(i\) from~\(U\), it makes up for the agent's score in the partitions that agent~\(i\) appears in as a candidate, which are pairwise disjoint for the agents in~\(U\).
This leads to a lower bound of \((k/b)/k = 1/b\), stated in \Cref{lem:bound-nice-n-k}.
To ensure the existence of \(b\)~such sets, we require that any subgraph of \(H\) induced by \(k\)~vertices can be partitioned into \(b\)~many (internally) independent sets.
We call a balanced partition system whose corresponding hypergraph has this property \emph{robust}.
In terms of the \(b\)-regular dual graph~\(G \defas H^*\), the condition is equivalent to the existence of an edge coloring with \(b\)~colors for every subgraph induced by \(k\)~edges: the edges of any one color do not share a vertex, which corresponds to vertices not sharing a hyperedge in~\(H\).
By K\H{o}nig's line coloring theorem~\citep{koenig1916}, a sufficient condition for such a coloring to exist is that \(G\)~is bipartite.%
\footnote{More generally, robust partition systems could be derived from any \(b\)-regular graph that is class~\(1\)~\citep{vizing1964}.}
The proofs of \Cref{lem:partitions,lem:bound-nice-n-k} will formalize these ideas.

Bipartite and \(b\)-regular graphs of even order \(k\) and size \(n\) exist for all \(b = 2n/k\) with \(b \leq k/2\).
A simple construction is depicted in \Cref{fig:graph} and described by the following lemma.

\begin{lemma}\label{lem:graph}
    Let \(b, k, n \in \N\) with \(k' \defas k/2 \in \N\) and \(b = 2n/k \leq k'\).
    Then, \(G = (V, E)\) with \(V \defas \range{k}\) and \(E \defas \left\{ \left\{ i, k' + \left( (i + \ell) \bmod k' \right) \right\} ~\middle|~ (i, \ell) \in \range{k'} \times \range{b} \right\}\) is a \(b\)-regular bipartite graph of order \(k\) and size \(n\).
\end{lemma}

\begin{proof}
    The order is given by $|V|=k$.
    To see that $G$ is bipartite, let \(\{u, v\} \in E\).
    Then, \(\{u, v\} = \left\{i, j\right\}\) with \(j \defas k' + \left( (i + \ell) \bmod k' \right)\) for some \(i \in \range{k'}\) and \(\ell \in \range{b}\).
    It follows from
    \begin{equation}
        i < k' \leq k' + \left( (i + \ell) \bmod k' \right) = j\label{eq:bipartite}
    \end{equation}
    that \(u < k'\) if and only if \(v \geq k'\).
    Hence, a bipartition of the vertices is given by $V = V_1 \disjoint V_2$ with $V_1 = \range{k'}$ and $V_2 = \range{k} \setminus \range{k'}$.
    To prove that the size of $G$ is $n$, we first show that \(f \colon \range{k'} \times \range{b} \to V \times V\) with \[f(i, \ell) \defas \left( i, k' + \left( (i + \ell) \bmod k' \right) \right)\] is injective.
    To this end, let \(f(i, \ell) = f(i', \ell')\) for some \((i, \ell), (i', \ell') \in \range{k'} \times \range{b}\).
    Clearly, it is \(i = i'\), so it remains to show that \(\ell = \ell'\).
    This is the case as
    \begin{align}
        && k' + \left( (i + \ell) \bmod k' \right) &= k' + \left( (i + \ell') \bmod k' \right)\label{eq:injective-projection} \\
        \Longleftrightarrow~&& i + \ell &= i + \ell' \pmod{k'}\nonumber{} \\
        \Longleftrightarrow~&& \ell &= \ell' \pmod{k'}\nonumber{}
    \end{align}
    which holds since \(0 \leq \ell, \ell' < b \leq k'\).
    From inequality~\eqref{eq:bipartite}, it follows that also \(f' \colon \range{k'} \times \range{b} \to E\) with \(f'(i, \ell) = \{{f(i, \ell)}_1, {f(i, \ell)}_2\}\) is injective, so \(|E| = k' b = kb/2 = n\) as required.
    For the degree, we consider each bipartition $V_1$ and $V_2$ separately.
    For \(i \in V_1\), let \(f_i \colon \range{b} \to \N\) with \(f_i(\ell) \defas {f(i, \ell)}_2\) enumerate the neighbors of vertex~\(i\) and assume \(f_i(\ell) = f_i(\ell')\) for some \(\ell, \ell' \in \range{b}\).
    As this implies equation~\eqref{eq:injective-projection}, we have again that \(\ell = \ell'\), so also \(f_i\) is injective and \(\deg(i) = b\).
    For \(j \in V_2\), the degree of vertex~\(j\) is \(\deg(j) = \left| \{(i, \ell) \in \range{k'} \times \range{b} \mid f_i(\ell) = j \} \right|\).
    It is
    \begin{align}
        && f_i(\ell) &= j \nonumber{} \\
        \Longleftrightarrow~&&
        (i + \ell) \bmod k' &= j - k' \nonumber{} \\
        \Longleftrightarrow~&&
        i &= j - k' - \ell \pmod{k'} \label{eq:rhs-regularity}
    \end{align}
    where the last equivalence follows from \(0 \leq j - k' < k'\).
    \Cref{eq:rhs-regularity} has a unique solution \(i \in \range{k'}\) for any fixed \(\ell \in \range{b}\), so \(\deg(j) = b\) and \(G\) is \(b\)-regular.
\end{proof}

\begin{figure}[t]
    \centering
    \foreach \b in {1,...,4}{
        \begin{subfigure}{\textwidth/4 - 4pt} 
            \centering
            \tikzmath{\k = 8;}
            \begin{tikzpicture}
    \input{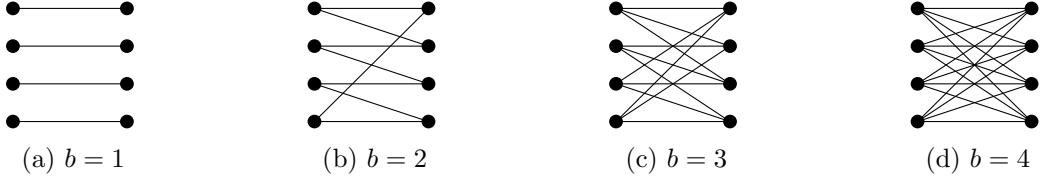}

    \tikzmath{
        coordinate \c, \d;
        \w = 1.5;
        \h = 0.5;
        \l = \k / 2;
        for \row in {0,...,\l-1} {
            for \col in {0,1} {
                \x = \col * \w;
                \y = \row * \h;
                \c = (\x, -\y);
                {
                    \node[vertex] at (\c) {};
                };
            };
        };
        for \row in {0,...,\l-1} {
            for \offset in {0,...,\b-1} {
                \y2 = Mod(\row + \offset, \l) * \h;
                \c = (0, -\row * \h);
                \d = (\w, -\y2);
                {
                    \draw[edge] (\c) -- (\d);
                };
            };
        };
    }
\end{tikzpicture}
            \caption{\(b = \b\)}%
        \end{subfigure}
    }
    \caption{%
        The construction of \Cref{lem:graph} for \(k = 8\) vertices and degree \(b \in [4]\): (a)~the \(4P_2\) (\(n = 4\) edges), (b)~the cycle \(C_8\) (\(n = 8\)), (c)~the cube graph \(Q_3\) (\(n = 12\)), and (d)~the complete bipartite graph \(K_{4, 4}\) (\(n = 16\)).
        Every edge represents an agent and every vertex corresponds to a partition.
        A vertex and an edge are incident if the corresponding agent is in the corresponding candidate set.
    }%
    \label{fig:graph}
    \Description{%
        Panel (a) shows two columns of four vertices aligned next to each other.
        Every vertex on the left is connected to its corresponding vertex on the right, totaling eight edges.
        The graph is labeled b = 1.
        Panel (b) is a copy of panel (a) that adds four more edges: the i-th vertex on the left is connected to the (i + 1)-th vertex on the right, modulo four.
        The resulting graph is the cycle on eight vertices; it is labeled b = 2.
        Panels (c) and (d) add four more edges each, connecting the i-th vertex on the left to the (i + 2)-th and to the (i + 3)-th vertex on the right, modulo four, respectively.
        The resulting graphs are the cube graph and the complete bipartite graph on eight vertices.
        They are labeled b = 3 and b = 4.
    }
\end{figure}

We condense the findings of this section in the following lemma.

\begin{lemma}\label{lem:partitions}
    Let \(n, k \in \N\) with \(k < n\) be such that \(b \defas 2n/k \in \N\) and \(b \leq k/2 \in \N\).
    Let further \(V\) with \(|V| = n\) denote a set of agents.
    Then, one may form \(k\) partitions \(S_1^p \disjoint S_2^p = V\), \(p \in [k]\), such that
    \begin{enumerate}[label=(\roman*)] 
        \item \(|S_2^p| = b\) for all \(p \in [k]\),\label{lem:partitions-1}
        \item \(|S_2^p \cap S_2^q| \leq 1\) for all \(p, q \in [k]\) with \(p \neq q\),\label{lem:partitions-2}
        \item \(|\{p \in [k] \mid v \in S_2^p\}| = 2\) for all \(v \in V\), and\label{lem:partitions-3}
        \item for every \(U \subseteq V\), there is a partition \({\dot\bigcup}_{\indx \in [b]} U_\indx = U\) with \(u \in S_2^p \Rightarrow v \not\in S_2^p\) for all \(\indx \in [b]\), \(u, v \in U_\indx\) with \(u \neq v\), and \(p \in [k]\).\label{lem:partitions-4}
    \end{enumerate}
\end{lemma}

\begin{proof}
    For \(n\), \(k\), and \(b\) as in the statement, \cref{lem:graph} guarantees the existence of a \(b\)-regular bipartite graph \(G = (X, V)\) of order \(|X| = k\) and size \(|V| = n\).
    Let \(H \defas G^* = (V, E)\) be its dual graph.
    Note that \(H\) is \(b\)-uniform and has order~\(n\) and size~\(k\).
    By \Cref{lem:simple_dual}, \(H\)~is further \(2\)-regular and linear.
    As \(b \geq 2\) by definition, it follows from linearity that \(H\) has no repeated edges, \ie{} \(E\) is a set.

    We use \(H\) to form a system of partitions of \(V\).
    First, enumerate \(E\) by an arbitrary but fixed bijection \(\phi \colon [k] \to E\).
    Then, for every \(p \in [k]\), define a candidate set \(S_2^p \defas \phi(p)\) and the associated voter set \(S_1^p \defas V \setminus \phi(p)\).
    As \(H\) is \(b\)-uniform, we have~\ref{lem:partitions-1} by construction.
    As it is linear,~\ref{lem:partitions-2}~follows.
    Since \(H\) is \(2\)-regular, also~\ref{lem:partitions-3}~holds.

    It remains to show property~\ref{lem:partitions-4}.
    By K\H{o}nig's line coloring theorem~\citep{koenig1916}, there exists a feasible edge \(b\)-coloring \(\pi \colon V \rightarrow [b]\) of \(G\).
    Let \(G'\) be the subgraph of \(G\) induced by an edge set \(U \subseteq V\).
    Clearly, \(\pi\)~restricted to~\(U\) remains a feasible edge \(b\)-coloring.
    The dual \(H' \defas {(G')}^*\) is the subgraph of~\(H\) induced by the vertex set~\(U\).
    In terms of~\(H'\), \(\pi\)~assigns colors to vertices.
    Since \(\pi\)~restricted to~\(U\) is feasible for~\(G'\), it follows from vertex--edge duality that vertices in~\(H'\) are colored differently if they appear in a hyperedge together, \ie{} \(\pi\) is a feasible vertex coloring for~\(H\). 
    Define thus \(U_\indx \defas \{v \in U \mid \pi(v) = \indx\}\) for each color \(\indx \in [b]\).
    Then, the sets \(U_\indx\) are disjoint by definition and \({\dot\bigcup}_{\indx \in [b]} U_\indx = U\) as \(\pi(U) \subseteq \pi(V) \subseteq [b]\).
    Let finally \(\indx \in [b]\) and \(u, v \in U_\indx\) with \(u \neq v\) and assume towards a contradiction that \(u, v \in S^p_2\) for some \(p \in [k]\).
    Then, \(u, v \in \phi(p) \in E\) and \(\pi(u) = \indx = \pi(v)\)
    by construction of \(S^p_2\) and \(U_\indx\), contradicting that \(\pi\) is a feasible vertex coloring for~\(H = (V, E)\).
\end{proof}

Formally, we write \(\mathcal{S}(n, k)\) for an arbitrary but fixed sequence \({\left((S_1^p, S_2^p)\right)}_{p \in [k]}\), with \(S_1^p \disjoint S_2^p = [n]\) for every \(p\in [k]\), that fulfills the conditions of \Cref{lem:partitions}.
We assume further that \(S_2^1 = [b]\).

\section{Impartial Selection}\label{sec:selection}

We are prepared to construct a mechanism that provides the first approximation guarantee for deterministic impartial selection with weighted votes.
Our main result is the following.

\begin{figure}[t]
    \centering
    \includegraphics[width=0.85\textwidth]{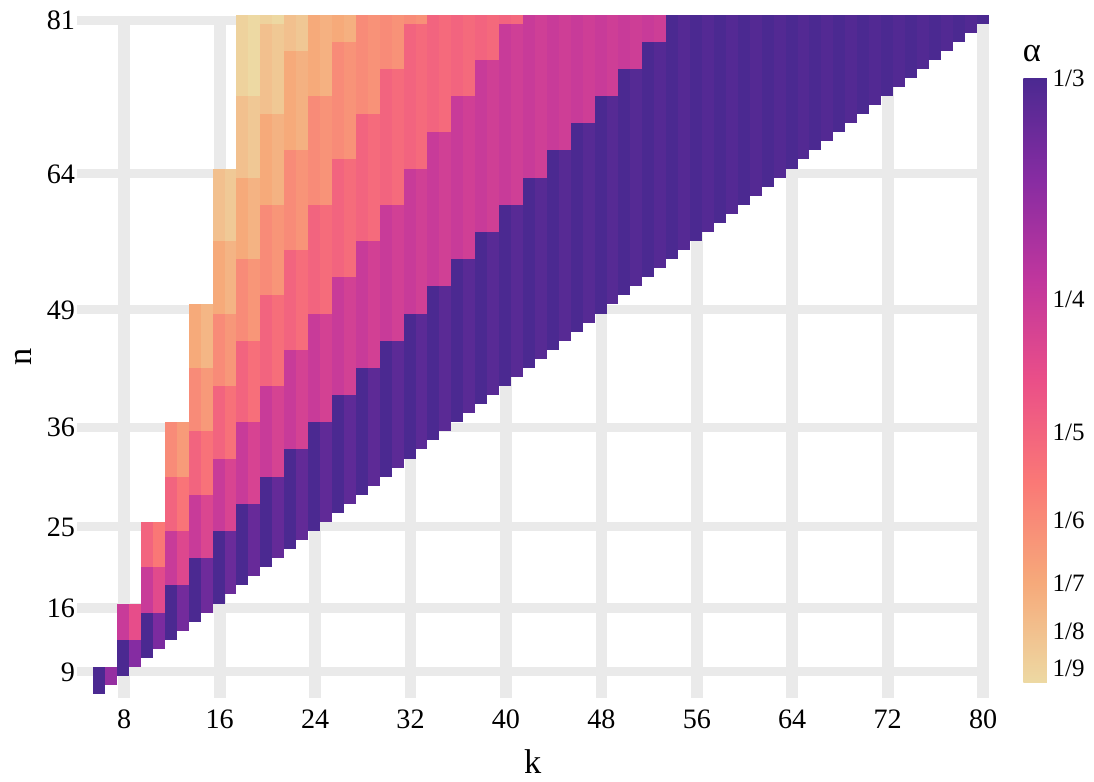}
    \caption{The performance guarantee of \Cref{thm:main} for permissible \(n\) and \(k\).}%
    \label{fig:rainbow}
    \Description{%
        The heatmap-style figure shows the performance guarantee of the main theorem for k from 6 to 80 on the horizontal axis and for n from 7 to 81 on the vertical axis.
        Permissible combinations of n and k are drawn in a color corresponding to the guarantee, roughly ranging from 1/9 to 1/3, while unsupported combinations are left transparent, revealing the white background and grid lines.
    }
\end{figure}

\begin{theorem}\label{thm:main}
    Let $n, k \in \N$ with $1 < k < n$ and $k-k\bmod 2 \geq 2\sqrt{n}$.
    Then, there exists an $(n,k)$-selection mechanism that is impartial and $\alpha$-optimal with
    \[
        \alpha = \frac{k-k\bmod 2}{k\left\lceil \frac{2n}{k-k\bmod 2} \right\rceil}.
    \]
\end{theorem}

The performance guarantee of \Cref{thm:main} is shown in \Cref{fig:rainbow}.
It starts from $2/k$ for $k-k\bmod2 = 2\sqrt{n}$ and grows up to $1/3$ for $k-k\bmod 2 \in [2n/3,n-1]$.

The main idea of the algorithm is as follows.
We construct a robust partition system of the set of agents, \ie{} a set of $k$ many partitions of the agents into voters and candidates such that each agent appears as a candidate twice and with disjoint sets of contenders.
For the second candidacy, we remove votes that are already present in the first candidacy to avoid double-counting.
Then, the mechanism selects the top scoring candidate from each partition, possibly selecting some agents twice.
This mechanism is impartial as voters and candidates are disjoint in each partition.
The performance guarantee stems mainly from the fact that every vote is counted exactly once.

In \Cref{sec:partitions}, we showed that a robust partition system is guaranteed to exist as long as $n$~and~$k$ satisfy $k<n$, $b\defas 2n/k\in \N$, and $b\leq k/2 \in \N$.
In the following, we assume these conditions in order to define and analyze our mechanism; we lift them in the end to obtain the general result stated in \Cref{thm:main}.

Given $n$ and $k$ as in \Cref{lem:partitions}, our selection mechanism is formally described by \Cref{alg:partitions-n-k}; we refer to it as $\algs_k$ and denote its output by $\algs_{k}(A)$ for a given input matrix $A\in \calA_n$.
\begin{algorithm}[t]
    \SetAlgoNoLine{}
    \KwIn{weight matrix $A\in \calA_n$}
    \KwOut{set $X\subseteq [n]$ with $|X|\leq k$}
    let $((S^1_1,S^1_2),\ldots,(S^k_1,S^k_2)) =\calS(n,k)$\;
    \For{$j\in [n]$}{
        let $\{l(j), r(j)\} = \{p \in [k]: j\in S^{p}_2\}$ with $l(j)<r(j)$\;
        define $\hat{\score}_{S^{l(j)}_1}(j) \xleftarrow{} \score_{S^{l(j)}_1}(j)$ and $\hat{\score}_{S^{r(j)}_1}(j) \xleftarrow{} \score_{S^{r(j)}_1 \setminus S^{l(j)}_1}(j)$\;
    }
    initialize $X\xleftarrow{} \emptyset$\;
    \For{$p\in [k]$}{
        take $i^p = \argmax_{j\in S^p_2} \, (\hat{\score}_{S^p_1}(j), j)$ and update $X\xleftarrow{} X\cup \{i^p\}$
        }
    {\bf return} $X$
    \caption{$\algs_{k}(A)$}%
    \label{alg:partitions-n-k}
\end{algorithm}
The procedure considers a partition system with the properties stated in \Cref{lem:partitions} and performs two main steps.
Recall that each agent $j \in [n]$ appears in two candidate sets; we denote their indices by $l(j) < r(j) \in [k]$ such that $j \in S^{l(j)}_2 \cap S^{r(j)}_2$.
The mechanism first computes the \emph{modified score} $\hat{\score}_{S^{p}_1}(j)$ for each $j\in [n]$ and each $p\in \{l(j),r(j)\}$, which is simply the actual score $\score_{S^{l(j)}_1}(j)$ for $p=l(j)$.
For $p=r(j)$, however, we omit the votes from agents $\smash{i\in S^{l(j)}_1}$ in order to avoid double counting.
The mechanism then selects the agent $i^p$ with the highest modified score out of each candidate set $S_2^p$, breaking ties in favor of the largest index.\footnote{We sometimes compare tuples, for example $(\score(j),j)$, in lexicographical order.
We use standard inequality signs as well as the $\min$ and $\max$ operators for this purpose.}
\Cref{fig:example-mechanism} illustrates a possible execution of $\algs_6$ on an instance $A\in \calA_{9}$.

\begin{figure}[t]
    \centering
    \begin{subfigure}{.4\textwidth}
        \centering
        \begin{align*}
            A=\begin{bmatrix}
                \cola{} & \colc{} & \cola{} & \cola{} & \cola{} & \cold{} & \cola{} & \cola{} & \colb{}\\ 
                \cola{} & \cola{} & \colb{} & \cola{} & \colb{} & \cola{} & \cola{} & \cola{} & \cola{}\\ 
                \cola{} & \cola{} & \cola{} & \cola{} & \cola{} & \cola{} & \cola{} & \cold{} & \cola{}\\ 
                \cola{} & \cola{} & \colc{} & \cola{} & \cola{} & \cola{} & \cola{} & \cola{} & \cola{}\\ 
                \colc{} & \cola{} & \cola{} & \cola{} & \cola{} & \cola{} & \cola{} & \cola{} & \cola{}\\ 
                \cola{} & \cola{} & \cola{} & \cola{} & \cold{} & \cola{} & \colb{} & \cola{} & \cola{}\\ 
                \cola{} & \cold{} & \colc{} & \cola{} & \cola{} & \cola{} & \cola{} & \colc{} & \cola{}\\ 
                \cola{} & \cola{} & \cola{} & \cola{} & \cola{} & \cola{} & \cola{} & \cola{} & \cola{}\\ 
                \cola{} & \cola{} & \cola{} & \colc{} & \cola{} & \colb{} & \cola{} & \colc{} & \cola{}   
            \end{bmatrix}
        \end{align*}
    \end{subfigure}
    \begin{subfigure}{.4\textwidth}
        \centering
        \includegraphics[width=0.7\textwidth,page=2]{example.pdf}
    \end{subfigure}
    \begin{subfigure}{.9\textwidth}
        \vspace{4mm}
        \includegraphics[width=\textwidth,page=1]{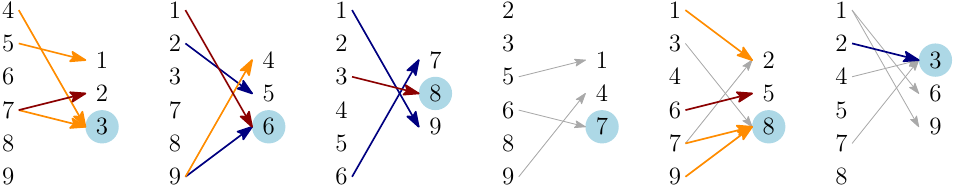}
    \end{subfigure}
    \caption{Example of \( \algs_6(A) \) for $A\in \calA_{9}$.
    The weight matrix~$A$ is shown alongside its graph representation, where votes of weight~1 are in blue, weight~2 are in orange, weight~3 are in red, and votes of weight~0 are not included.
    The partition system is given below, where omitted votes are shown in gray.
    For each partition, the selected agent is highlighted in light blue.
    Observe that \(\score(\algs_6(A)) = 17\) and \( \score(\opt_6(A)) = 27 \); the multiplicative guarantee provided by \Cref{lem:bound-nice-n-k} for this instance is $1/3$.}%
    \label{fig:example-mechanism}
    \Description{%
        The top left panels shows the nine times nine poll matrix A, where the entry in the i-th row and j-th column denotes the weight of the vote cast from agent i to agent j.
        The agents vote as follows.
        Agent 1 votes for agents 2, 6, and 9 with weights 2, 3, and 1.
        Agent 2 for agents 3 and 5, each with weight 1.
        Agent 3 for agent 8 with weight 3.
        Agent 4 for agent 3 with weight 2.
        Agent 5 for agent 1 with weight 2.
        Agent 6 for agents 5 and 7 with weights 3 and 1.
        Agent 7 for agents 2, 3, and 8 with weights 3, 2, and 2.
        Agent 8 casts no votes.
        Agent 9 votes for agents 4, 6, and 8 with weights 2, 1, and 2.
        All other entries of the matrix are zero.
        The top right panel shows this matrix interpreted as the adjacency matrix of a digraph, with arcs colored according to their weight and with zero-weight arcs omitted.
        The bottom panel shows six partitions of the nine agents.
        The left side of each partition contains six and the right side contains three agents.
        The agents that appear on the right sides are, in order, {1, 2, 3}, {4, 5, 6}, {7, 8, 9}, {1, 4, 7}, {2, 5, 8}, and {3, 6, 9}.
        Only arcs going from the left side to the right side are included in each partition.
        Arcs that appear for a second time are grayed out.
        On each right side, one agent is highlighted as the winner.
        These are, in order, agents 3, 6, 8, 7, 8, and 3.
    }
\end{figure}

In this section, whenever $n,~k$, and $A\in \calA_n$ are fixed, we write $((S^1_1,S^1_2), \ldots, (S^k_1,S^k_2))$, $l(j)$, $r(j)$, $\hat{\score}_{S^p_1}(j)$, $i^p$, and $X$ for each $p\in [k]$ and $j\in [n]$ to refer to the objects defined in $\algs_k$.
We only specify the input matrix $A$ as an argument when it is not clear from the context.
The following lemma constitutes the main technical ingredient for the proof of \Cref{thm:main}.

\begin{lemma}\label{lem:bound-nice-n-k}
    Let $n,k\in \N$ with $k<n$ be such that  $b\defas 2n/k \in \N$ and $b\leq k/2 \in \N$.
    Then, $\algs_{k}$ is an impartial and $1/b$-optimal $(n,k)$-selection mechanism.
\end{lemma}

\begin{proof}
    We consider $n$ and $k$ as in the statement.
    We first note that $\algs_{k}$ returns a subset of~$[n]$ of size at most~$k$ and is well-defined as we have $|\{p \in [k]: j\in S^{p}_2\}| = 2$ for every $j\in [n]$.
    The former holds since $i^p$ is a single agent for every $p\in [k]$ and $X=\bigcup_{p\in [k]}\{i^p\}$; the latter follows from property~\ref{lem:partitions-3} of \Cref{lem:partitions} since
    $b\defas 2n/k \in \N$ and $b\leq k/2 \in \N$.

    To see that $\algs_{k}$ is impartial, let $A, A'\in \calA_n$ and $j\in [n]$ such that $A_{-j} = A'_{-j}$.
    Suppose $j\in \algs_{k}(A)$.
    From the definition of the mechanism, we have that there is a $p\in [k]$ such that $j = \argmax_{i\in S^p_2} (\hat{\score}_{S^p_1}(i;A), i)$.
    Since $j\in S^p_2$ and $A_{-j}=A'_{-j}$, we have both that $\hat{\score}_{S^p_1}(j;A) = \hat{\score}_{S^p_1}(j;A')$ and, for every $i\in S^p_2\setminus \{j\}$, that $\hat{\score}_{S^p_1}(i;A) = \hat{\score}_{S^p_1}(i;A')$.
    This yields $j = \argmax_{i\in S^p_2} (\hat{\score}_{S^p_1}(i;A'), i)$.
    Thus, we obtain from the definition of the mechanism that $j\in \algs_{k}(A')$.
    Exchanging the roles of $A$~and~$A'$ in the previous argument, we obtain that if $j\in \algs_{k}(A')$, then also $j\in \algs_{k}(A)$.
    We thus conclude that $\algs_{k}(A)\cap \{j\} = \algs_{k}(A')\cap \{j\}$.

    It remains to show that $\algs_k$ has an approximation ratio of $1/b$.
    To this end, we let $A\in \calA_n$ be an arbitrary weight matrix.
    First, observe that
    \begin{equation}
        \hat{\score}_{S^{l(j)}_1}(j) + \hat{\score}_{S^{r(j)}_1}(j) = \score_{S^{l(j)}_1}(j) + \score_{S^{r(j)}_1 \setminus S^{l(j)}_1}(j) = \score(j) \label{eq:sum-modified-indegrees}
    \end{equation}
    for every $j\in [n]$, since property~\ref{lem:partitions-2} of \Cref{lem:partitions} implies $S^{l(j)}_1\cup S^{r(j)}_1 = [n]\setminus \{j\}$.
    Furthermore, the definition of $i^p$ yields that
    \begin{equation}\label{eq:selected-vertex}
        \hat{\score}_{S^p_1}(i^p) \geq \hat{\score}_{S^p_1}(j)
    \end{equation}
    for every $p\in [k]$ and $j\in S^p_2$.
    Given these two facts, we claim that
    \begin{equation}
       \hat{\score}_{S^{l(j)}_1}(i^{l(j)}) + \hat{\score}_{S^{r(j)}_1}(i^{r(j)}) \geq \score(j) \label{eq:indegree-two-sets}
    \end{equation}
    for every $j\in [n]$.
    To see this, we fix $j\in [n]$.
    If $i^p=j$ for each $p\in \{l(j),r(j)\}$, inequality~\eqref{eq:indegree-two-sets} follows immediately from equality~\eqref{eq:sum-modified-indegrees}.
    If $|\{j\}\cap \{i^p: p\in \{l(j),r(j)\}\}| = 1$, say w.l.o.g.\ $i^{l(j)}=j$ and $i^{r(j)}=h \not= j$, we have that
    \[
        \hat{\score}_{S^{r(j)}_1}(h) \geq \hat{\score}_{S^{r(j)}_1}(j) = \score(j) - \hat{\score}_{S^{l(j)}_1}(j),
    \]
    where the inequality follows from~\eqref{eq:selected-vertex} and the equality from~\eqref{eq:sum-modified-indegrees}.
    In this case, inequality~\eqref{eq:indegree-two-sets} follows from $j=i^{l(j)}$ and $h=i^{r(j)}$. Finally, if $j\not\in \{i^p: p\in \{l(j),r(j)\}\}$, we have from~\eqref{eq:selected-vertex} that
    \[
        \hat{\score}_{S^{l(j)}_1}(i^{l(j)}) \geq \hat{\score}_{S^{l(j)}_1}(j)\quad
        \text{and}
        \quad
        \hat{\score}_{S^{r(j)}_1}(i^{r(j)}) \geq \hat{\score}_{S^{r(j)}_1}(j)
    \]
    so that inequality~\eqref{eq:indegree-two-sets} follows from summing up these two inequalities and applying equality~\eqref{eq:sum-modified-indegrees}.
    This concludes the proof of inequality~\eqref{eq:indegree-two-sets}.

    Letting $\chi$ denote the indicator function for logical propositions, we note that
    \begin{align}
        \score(\algs_{k}(A))
        & = \sum_{j\in \algs_k(A)} \score(j) \nonumber \\
        & = \sum_{j\in \algs_k(A)} \left( \hat{\score}_{S^{l(j)}_1}(j) + \hat{\score}_{S^{r(j)}_1}(j) \right) \nonumber \\
        & \geq \sum_{j\in \algs_k(A)} \left( \hat{\score}_{S^{l(j)}_1}(j) \chi(j=i^{l(j)}) + \hat{\score}_{S^{r(j)}_1}(j) \chi(j=i^{r(j)}) \right) \nonumber \\
        & = \sum_{p\in [k]} \hat{\score}_{S^p_1}(i^p). \label{eq:indegree-alg}
    \end{align}
    Indeed, the first equality follows from the definition of $\algs_k(A)$, the second one from equality~\eqref{eq:sum-modified-indegrees}, the inequality simply from $\chi(\cdot) \leq 1$, and the last equality follows from the definition of $i^p$ for each $p\in [k]$.
    We next use inequalities\ \eqref{eq:indegree-two-sets}~and~\eqref{eq:indegree-alg} to conclude the bound stated in the lemma.

    For $b\defas 2n/k$, we know from property~\ref{lem:partitions-4} of \Cref{lem:partitions} that there is a partition \({\dot\bigcup}_{\indx \in [b]} U_{\indx} = \opt_k(A)\) such that \(i \in S_2^p\) implies \(j \not\in S_2^p\) for all \(\indx \in [b]\), \(i, j \in U_{\indx}\) with $i\not=j$, and \(p \in [k]\). We obtain that, for every \(\indx \in [b]\),
    \begin{equation}
        \score(\algs_{k}(A)) \geq \sum_{p\in [k]} \hat{\score}_{S^p_1}(i^p) \geq \sum_{j\in U_{\indx}} \left( \hat{\score}_{S^{l(j)}_1}(i^{l(j)}) + \hat{\score}_{S^{r(j)}_1}(i^{r(j)}) \right) \geq \score(U_{\indx}),\label{eq:indegree-clique}
    \end{equation}
    where the first inequality follows from inequality~\eqref{eq:indegree-alg}, the second one from the fact that $\{l(i),r(i)\}\cap \{l(j),r(j)\}=\emptyset$ for every $\indx\in [b]$ and every $i,j\in U_{\indx}$ with $i\not= j$, and the last one from inequality~\eqref{eq:indegree-two-sets}.
    This yields
    \[
        \score(\algs_{k}(A)) \geq \max_{\indx\in [b]} \score(U_\indx) \geq \frac{1}{b} \sum_{\indx\in [b]} \score(U_{\indx}) = \frac{1}{b}\score(\opt_k(A)).
    \]
    Here, the first inequality follows from~\eqref{eq:indegree-clique}, the second one from the observation that the maximum of a set of values is at least as large as their average, and the equality from the fact that \( {\{U_\indx\}}_{\indx\in [b]} \) is a partition of $\opt_k(A)$.
    Therefore, we obtain that \(\algs_{k}\) is $\alpha$-optimal for
    \[
         \frac{\score(\algs_{k}(A))}{\score(\opt_k(A))} \geq \frac{1}{b} = \alpha. \qedhere
    \]
\end{proof}

In order to conclude our main result, it only remains to extend the bound given by \Cref{lem:bound-nice-n-k} to the case
where at least one of the conditions $b\defas 2n/k \in \N$ or $b\leq k/2 \in \N$ is not satisfied.
To this end, we show a general way to extend bounds on the approximation ratio for given values of $\tilde{n}$~and~$\tilde{k}$ to other values $n$~and~$k$: whenever $n\leq \tilde{n}$ and $k\geq \tilde{k}$, we can do so preserving impartiality and only losing a factor of~$\tilde{k}/k$.

Given $k, \tilde{k}, \tilde{n}, n\in \N$ with $\tilde{k} \leq k< n\leq \tilde{n}$, and an $(\tilde{n},\tilde{k})$-selection mechanism $\alg$, we can generalize $\alg$ to the $(n,k)$-selection mechanism $\gen_{\alg,k}$.
This is formally described by \Cref{alg:extension-alg-k}, whose output is denoted by $\gen_{\alg,k}(A)$ for an input matrix $A\in \calA_n$.
\begin{algorithm}[t]
    \SetAlgoNoLine{}
    \KwIn{weight matrix $A\in \calA_n$}
    \KwOut{set $X\subseteq [n]$ with $|X|\leq k$}
    define $\tilde{A}\in \calA_{\tilde{n}}$ as
    \[
        \tilde{A}_{ij} = \begin{cases}
            A_{ij} & \text{ if } i,j\in [n],\\
            0 & \text{ otherwise.}
        \end{cases}
    \]
    {\bf return} $\alg(\tilde{A})$
    \caption{$\gen_{\alg,k}(A)$}%
    \label{alg:extension-alg-k}
\end{algorithm}
This algorithm simply extends $A$ to the $\tilde{n}\times\tilde{n}$ matrix~$\tilde{A}$ by adding $\tilde{n}-n$ many all-zero rows and columns to it, and then applies $\alg$ on $\tilde{A}$.
As before, whenever $\tilde{n},~n,~k,~\alg$, and $A\in \calA_n$ are fixed, we use $\tilde{A}$ to refer to the object defined in \Cref{alg:extension-alg-k} for this input.
In a slight overload of notation, when we consider $A'\in \calA_n$ as an input, we write simply $\tilde{A}'$ for the matrix defined in \Cref{alg:extension-alg-k} on input $A'$.
We obtain the following lemma.

\begin{lemma}\label{lem:bounds-nice-not-nice-n-k}
    Let $\tilde{k},k,n,\tilde{n} \in \N$ with $\tilde{k}\leq k< n\leq \tilde {n}$ be such that there exists an impartial and $\tilde{\alpha}$-optimal $(\tilde{n},\tilde{k})$-selection mechanism $\alg$.
    Then $\gen_{\alg,k}$ is an impartial and $\alpha$-optimal $(n,k)$-selection mechanism with
    $ \alpha = (\tilde{k}/k) \tilde{\alpha}$.
\end{lemma}

\begin{proof}
    Let $n$, $k$, $\tilde{n}$, and $\tilde{k}$ be as in the statement.
    Let also $\alg$ denote the impartial and $\tilde{\alpha}$-optimal $(\tilde{n},\tilde{k})$-selection mechanism.
    In order to see that $\gen_{\alg,k}$ is impartial, let $A,A'\in \calA_n$ and $i\in [n]$ such that $A_{-i}=A'_{-i}$.
    This implies $\tilde{A}_{-i}=\tilde{A}'_{-i}$, thus the impartiality of $\alg$ yields
    \[
        \gen_{\alg,k}(A) \cap \{i\} = \alg(\tilde{A}) \cap \{i\} = \alg(\tilde{A}') \cap \{i\} = \gen_{\alg,k}(A') \cap \{i\}.
    \]

    To prove the approximation guarantee, we  let $A\in \calA_n$ be an arbitrary weight matrix and observe that
    \begin{equation}\label{eq:alpha-tilde1}
        \frac{\sigma(\gen_{\alg,k}(A))}{\sigma(\opt_{\tilde{k}}(\tilde{A}))} = \frac{\sigma(\alg(\tilde{A}))}{\sigma(\opt_{\tilde{k}}(\tilde{A}))} \geq \tilde{\alpha},
    \end{equation}
    where the equality follows from the definition of $\gen_{\alg,k}$ and the inequality follows from the $\tilde{\alpha}$-optimality of $\alg$.
    On the other hand, as $\tilde{k}\leq k$ and $\sigma(j,\tilde{A})=0$ for every $j\not \in [n]$, we know that
    \[
        \frac{\sigma(\opt_{k}(A))}{k} = \frac{1}{k} \max_{S \subseteq [n] \colon |S| = k} \score(S;A) \leq \frac{1}{\tilde{k}} \max_{S \subseteq [n] \colon |S| = \tilde{k}} \score(S;A)= \frac{\sigma(\opt_{\tilde{k}}(\tilde{A}))}{\tilde{k}},
    \]
    \ie{} the average score of the $k$ top-voted agents of input $A$ can be no larger than the average score of the $\tilde{k}$ top-voted agents of input $\tilde{A}$.
    Plugging this inequality into~\eqref{eq:alpha-tilde1} concludes the proof as
    \[
        \frac{\sigma(\gen_{\alg,k}(A))}{\sigma(\opt_k(A))} \geq \frac{\tilde{k}}{k} \frac{\sigma(\gen_{\alg,k}(A))}{\sigma(\opt_{\tilde{k}}(\tilde{A}))} \geq \frac{\tilde{k}}{k} \tilde{\alpha}. \qedhere
    \]
\end{proof}

Our main result now follows from the last two lemmas.

\begin{proof}[Proof of \Cref{thm:main}]
    Let $n$ and $k$ be as in the statement.
    We define
    \[
        \tilde{k}\defas k-k\bmod 2 \quad \text{and} \quad \tilde{n} \defas \frac{k-k\bmod 2}{2} \left\lceil \frac{2n}{k-k\bmod 2} \right\rceil.
    \]
    It is clear that $\tilde{n},\tilde{k}$ are natural numbers with $\tilde{k}\leq k < n\leq \tilde{n}$ and that
    \[
        b\defas \frac{2\tilde{n}}{\tilde{k}} = \left\lceil \frac{2n}{k-k\bmod 2}\right\rceil\in \N.
    \]
    Moreover, we have that
    \[
        \tilde{n} = \frac{k-k\bmod 2}{2} \left\lceil \frac{2n}{k-k\bmod 2} \right\rceil \leq \frac{k-k\bmod 2}{2} \left\lceil \frac{2\frac{{(k-k\bmod2)}^2}{4}}{k-k\bmod 2} \right\rceil = \frac{\tilde{k}^2}{4},
    \]
    where the inequality follows from the condition $k-k\bmod 2\geq 2\sqrt{n}$ in the statement.
    This yields $b = 2\tilde{n}/\tilde{k} \leq \tilde{k}/2\in \N$.
    By \Cref{lem:bound-nice-n-k}, this implies that $\algs_{\tilde{k}}$ is an impartial and $\tilde{\alpha}$-optimal $(\tilde{n},\tilde{k})$-selection mechanism with
    \[
        \tilde{\alpha} = \frac{1}{b} = \frac{1}{\left\lceil \frac{2n}{k-k\bmod 2}\right\rceil}.
    \]
    Since $\tilde{n},\tilde{k}\in \N$ are such that $\tilde{k}\leq k$ and $\tilde{n}\geq n$, \Cref{lem:bounds-nice-not-nice-n-k} implies that $\gen_{\algs_{\tilde{k}},k}$ is an impartial and $\alpha$-optimal $(n,k)$-selection mechanism with
    \[
        \alpha = \frac{\tilde{k}}{k} \tilde{\alpha} = \frac{k-k\bmod 2}{k \left\lceil \frac{2n}{k-k\bmod 2}\right\rceil}. \qedhere
    \]
\end{proof}

The mechanism and its approximation ratio naturally extend to the widely studied unweighted setting, where one restricts to matrices $A\in \calA_n$ with $A_{ij}\in \{0,1\}$ for every $i,j\in [n]$.
This improves on the previous best lower bound of $1/k$ whenever the number of agents to select is high enough compared to $n$ for \Cref{thm:main} to be applicable: if $k-k\bmod 2 \geq 2\sqrt{n}$, the theorem guarantees the existence of an $(n,k)$-selection mechanism that is impartial and $\alpha$-optimal with
\[
    \alpha = \frac{k-k\bmod 2}{k\left\lceil \frac{2n}{k-k\bmod 2} \right\rceil} \geq \frac{k-k\bmod 2}{k\left\lceil \frac{2{(k-k\bmod2)}^2}{4(k-k\bmod 2)} \right\rceil} = \frac{2}{k}.
\]

We end this section by showing that the analysis of our $(n,k)$-selection mechanism $\algs_k$ for $n$ and $k$ satisfying the conditions of \Cref{lem:bound-nice-n-k} is tight.

\begin{theorem}\label{thm:tightness}
    Let $n,k\in\N$ with $k<n$ be such that $b\defas 2n/k \in \N$ and $b \leq k/2 \in \N$.
    Then, for every $\varepsilon>0$ we have that $\algs_k$ is not $(1/b+\varepsilon)$-optimal.
\end{theorem}

\begin{figure}[t]
    \centering
    \begin{subfigure}{\textwidth}
    \centering
        \includegraphics[width=.9\textwidth,page=3]{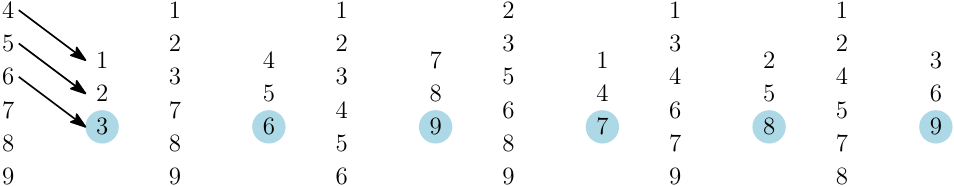}
    \end{subfigure}
    \caption{%
        Example of the construction of the proof of \Cref{thm:tightness} for $n=9$ and $k=6$ with $3$~votes of weight~$1$: agent~$4$ votes for agent~$1$, agent~$5$ votes for agent~$2$, and agent~$6$ votes for agent~$3$.
        All votes are only seen in the first partition.
        Since agents with positive scores have the smallest indices, they are not selected in their second candidate set.
    }%
    \label{fig:tightness}
    \Description{%
        The figure shows the same six partitions of nine agents as the previous figure, but only three weight-one votes as described by the caption are included.
        The winning agents are 3, 6, 9, 7, 8, and 9.
    }
\end{figure}

\begin{proof}
    Let $n$ and $k$ be as in the statement and consider the partition system \[((S^1_1,S^1_2),\ldots,(S^k_1,S^k_2)) =\calS(n,k).\]
    Recall that we defined $\calS(n,k)$ such that $S^1_2=[b]$.
    Considering $l(j)$ and $r(j)$ as defined in \Cref{alg:partitions-n-k} for every $j\in [n]$, we note that for each $j\in S^1_2$ we have $l(j)=1$.
    For each $j\in S^1_2$, we let $h(j)$ be an arbitrary agent in $S^1_1$ such that $h(j)\in S^{r(j)}_2$.
    Such agent is guaranteed to exist, since from property~\ref{lem:partitions-2} of \Cref{lem:partitions} we know that $S^{l(j)}_2 \cap S^{r(j)}_2 = \{j\}$, and from property~\ref{lem:partitions-1} we have that $|S^{r(j)}_2| = b > 1$.

    We consider the instance given by $A\in \calA_n$ with $A_{ij}=1$, if $j\in S^1_2$ and $i=h(j)$, and $A_{ij}=0$, otherwise.
    Intuitively, this construction aims to have $A_{ij}>0$ for some $i\in S^p_1$ and $j\in S^p_2$ only if $p=1$, so that the only agent with a strictly positive score selected by the mechanism, among $b$ agents with a strictly positive score, is $i^1$.
    An example of this construction and the corresponding outcome of the mechanism is illustrated in \Cref{fig:tightness}.

    It is clear that $\opt_k(A)=[b]$ and $\score(\opt_k(A)) = b$.
    On the other hand, we have that $\score(i^1)= 1$ and, for every $p\in \{2,3,\ldots,k\}$, that $\hat{\score}_{S^p_1}(j) = 0$ for every $j\in S^p_2$.
    This is because we have $\score(j)=0$ for every $j\not \in [b]$ and, whenever there is a $j\in [b] \cap S^2_p$, we also have $h(j)\in S^2_p$.
    Moreover, for every $p\in \{2,3,\ldots,k\}$ such that there exists a $j\in [b] \cap S^2_p$, we have that $j\not= \max S^p_2$ since $h(j)\in S^p_2$ and $h(j)>j$.
    This yields $\score(i^p)=0$ for every $p\in \{2,3,\ldots,k\}$, thus $\score(\algs_k(A))=1$.
    This concludes the proof as
    \[
        \frac{\score(\algs_k(A))}{\score(\opt_k(A))} = \frac{1}{b}. \qedhere
    \]
\end{proof}

In terms of general upper bounds on the approximation ratio that an impartial mechanism can achieve, the best known is $(k-1)/k$ \citep{bjelde2017impartial}.
Even for the regime $k-k\bmod 2\geq 2n/3$, in which our mechanism provides a lower bound of $1/3$ and considerably improves the previously best bound of $1/k$ \citep{bjelde2017impartial}, the gap remains large. Further improvements in either lower or upper bounds arise as the main direction for future work.

\section{Impartial Assignment}\label{sec:assignment}

In this section, we consider a generalization of the impartial selection problem in which agents are not selected into one but \emph{assigned} to at most one of $m$ many sets, which we refer to as \emph{jobs}.
Each job $\ell \in [m]$ can be assigned at most $k$ agents, so that we obtain the impartial selection problem as the special case where $m = 1$.
We first extend the notation from \Cref{sec:prelims} to this new setting.

For \(n, m\in \N\) with \(m\leq n\), we consider $m$-tuples of weight matrices \(\bfA = (A_1,A_2,\ldots,A_m)\in \calA^m_n\), each of them representing the weighted votes for one job.
Let further $k < n$ in the following; an instance of the assignment problem is then given by the tuple $\bfA$ and the value $k$.
We let
\begin{align*}
    \calX_k \defas \big\{ &\bfX=(X_1,X_2,\ldots,X_m)\in {\big(2^{[n]}\big)}^m\colon |X_i|\leq k \text{ and } X_i\cap X_j = \emptyset \\
    & \text{for every } i,j\in [m] \text{ with } i\not= j\big\}
\end{align*}
denote the set of feasible assignments, \ie{} the set of tuples $\bfX$ containing $m$ pairwise disjoint subsets of agents, each with cardinality at most $k$.
In a slight overload of notation, for $\bfX\in \calX_k$ and $\bfA\in \calA^m_n$, we write
\[
    \score(\bfX;\bfA) \defas \sum_{\ell\in[m]} \score(X_\ell;A_\ell)
\]
to refer to the sum, over the jobs, of the score of the set assigned to each job according to $\bfX$, and we simply write $\score(\bfX)$ when the instance is clear from the context.
Finally, for $\bfA \in \calA^m_n$, we let
\[
    \opt_k(\bfA) \defas \argmax_{\bfX \in \calX_k} \score(\bfX;\bfA)
\]
denote an arbitrary assignment with the largest score among feasible assignments.
We write just \(\opt_k\) when the instance is clear.

An $(n,m,k)$-assignment mechanism is a function $f \colon \calA^m_n \to {\big(2^{[n]}\big)}^m$ such that $f(\bfA)\in \calX_k$ for every $\bfA\in \calA^m_n$.
Such a mechanism is \emph{impartial} if, for every pair of instances $\bfA \in \calA^m_n$ and $\bfA'\in \calA^m_n$ and for all agents $i\in [n]$ such that ${(A_\ell)}_{-i}={(A'_\ell)}_{-i}$ holds for each job $\ell\in [m]$, it also holds that ${(f(\bfA))}_\ell \cap \{i\} = {(f(\bfA'))}_\ell \cap \{i\}$ for every $\ell\in [m]$.
We further call an $(n,m,k)$-assignment mechanism \emph{$\alpha$-optimal} if
\[
    \frac{\score(f(\bfA);\bfA)}{\score(\opt_k(\bfA);\bfA)} \geq \alpha
\]
holds for all $\bfA \in \calA^m_n$ and some $\alpha \in [0,1]$.

We are prepared to state the main theorem of this section.

\begin{theorem}\label{thm:assignment}
    Let $n,m,k\in \N$ with $1 < k < n,~ mk\leq 2n$, and $k-k\bmod 2 \geq 2\sqrt{n}$.
    Then, there exists an $(n,m,k)$-assignment mechanism that is impartial and $\alpha$-optimal with
    \[
        \alpha = \frac{k-k\bmod 2}{2k\left\lceil \frac{2n}{k-k\bmod 2} \right\rceil}.
    \]
\end{theorem}

The main ingredient of the proof is an adaptation of our mechanism from \cref{sec:selection} that selects from each partition not one but $m$~many agents: one for each set $\ell\in [m]$.
We leave the partitioning step unchanged and, for the second step, assign~$m$ agents from each candidate set to different jobs in a way that the score obtained for each partition is maximized.
In case an agent is assigned to two different jobs, we assign it to the one for which it receives the highest number of votes.
The adapted procedure is formally described in \Cref{alg:assign-k}; we refer to it as $\alga_k$ and denote its output by $\alga_{k}(\bfA)$ for a given input tuple of matrices $\bfA\in \calA^m_n$.

\begin{algorithm}[t]
    \SetAlgoNoLine{}
    \KwIn{tuple of weight matrices $\bfA\in \calA^m_n$}
    \KwOut{assignment $\bfX\in \calX_k$}
    let $((S^1_1,S^1_2),\ldots,(S^k_1,S^k_2)) =\calS(n,k)$\;
    \For{$j\in [n]$}{
        let $\{l(j), r(j)\} = \{p \in [k]: j\in S^{p}_2\}$ with $l(j)<r(j)$\;
        \For{$\ell\in [m]$}{
            define $\hat{\score}_{S^{l(j)}_1}(j;A_\ell) \xleftarrow{} \score_{S^{l(j)}_1}(j;A_\ell)$; \\
            define $\hat{\score}_{S^{r(j)}_1}(j;A_\ell) \xleftarrow{} \score_{S^{r(j)}_1 \setminus S^{l(j)}_1}(j;A_\ell)$
        }
    }
    initialize $X_\ell \xleftarrow{} \emptyset$ for each $\ell\in[m]$\;
    \For{$p\in [k]$}{
        take $\bfx^p = \argmax_{\substack{\bfv\in {(S^p_2)}^m:  v_\ell \not=  v_{\ell'} \forall \ell\not=\ell'}}\, \big(\sum_{\ell\in[m]} \hat{\score}_{S^p_1}( v_\ell;A_\ell), v_1, \ldots, v_m \big)$\;
        update $X_\ell\xleftarrow{} X_\ell\cup \{x^p_\ell\}$ for each $\ell\in [m]$
    }
    \For{$j\in [n]$}{
        \If{$L(j) \defas \{\ell\in [m]\colon j\in X_\ell\}$ is such that $|L(j)|= 2$}{
            let $\hat{\ell} = \argmin_{\ell\in L(j)}\,(\score(j;A_\ell),\ell)$\;
            $X_{\hat{\ell}} \leftarrow X_{\hat{\ell}} \setminus \{j\}$
            }
    }
    {\bf return} $\bfX$
    \caption{$\alga_{k}(\bfA)$}%
    \label{alg:assign-k}
\end{algorithm}

Impartiality of this mechanism follows from a similar reasoning as in the proof of \cref{thm:main}: whenever the vote of an agent is taken into account, the agent is not part of the candidate set.
The approximation guarantee makes use of a detailed analysis of the case $b\defas 2n/k\in \N$ and $b\leq k/2\in \N$, which is somewhat more intricate than the analysis in \Cref{sec:selection}.
We consider subsets of agents that are assigned to any job in the optimal assignment and are not mutual contenders.
We then use the key fact that, when considering the two partitions in which some agent~$i$ is in the candidate set, the mechanism assigns agents in a way that the sum of votes of the assigned agents in both partitions is at least the number of votes that $i$~receives for any job.
Exploiting the robust partitioning structure as before allows us to take the best of these subsets and conclude via an averaging argument.
Here we lose an additional factor of $1/2$ due to the possibility that an agent is initially assigned to two jobs.
The extension to general values $n,~m$, and $k$ is then analogous to that of \Cref{sec:selection}.
\Cref{fig:example-mechanism-assignment} illustrates a possible execution of $\alga_6$ on an instance $A\in \calA_9^2$.

\begin{figure}[t]
    \centering
    \begin{subfigure}{.35\textwidth}
        \includegraphics[height=44mm,page=6]{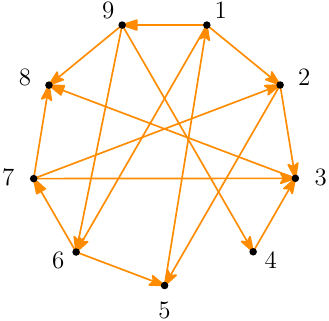}
    \end{subfigure}
    \hspace{17mm}
    \begin{subfigure}{.35\textwidth}
        \includegraphics[height=44mm,page=8]{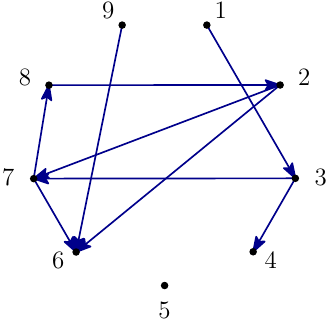}
    \end{subfigure}

    \begin{subfigure}{.9\textwidth}
        \includegraphics[width=\textwidth,page=5]{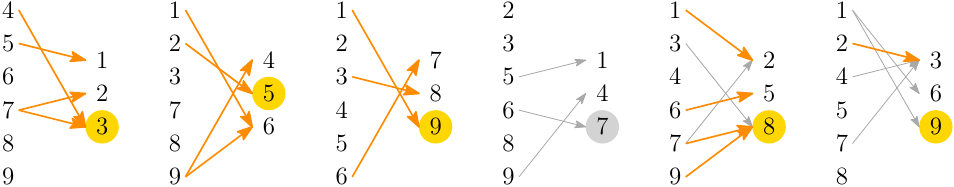}
    \end{subfigure}

    \vspace{3mm}
    \begin{subfigure}{.9\textwidth}
        \includegraphics[width=\textwidth,page=7]{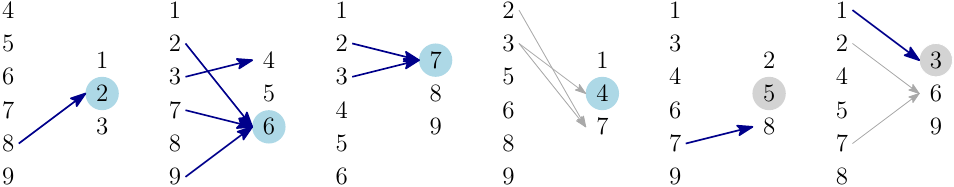}
    \end{subfigure}
    \caption{%
        Example of $\alga_6(\bfA)$ for $\bfA \in \calA_9^2$.
        The votes carry unit weight and are shown in orange for job~$1$ and in blue for job~$2$.
        The partition system is depicted below, separately for each job, with omitted votes drawn in gray.
        For each partition and each job, the selected agent is highlighted in light orange or light blue, if it remains assigned to the corresponding job, and in gray, if it is unassigned due to a duplicate preliminary assignment of the agent to distinct jobs.
        This affects agents $3$,~$5$, and~$7$:
        Since agent~$3$ has three votes for job~$1$ and only one vote for job~$2$, it is unassigned from job~$2$.
        Agent~$5$ has two votes for job~$1$ and no votes for job~$2$, so it is unassigned from job~$2$.
        Likewise, agent~$7$ has one vote for job~$1$ and two votes for job~$2$, and is thus unassigned from job~$1$.
        It can be confirmed that \(\score(\alga_6(\bfA)) = 16\) and \(\score(\opt_6(\bfA)) = 18\) (reassign agents $1$~and~$2$ to job~$1$); the multiplicative guarantee provided by \Cref{lem:bound-nice-n-k-assignment} for this instance is $1/6$.
    }%
    \label{fig:example-mechanism-assignment}
    \Description{%
        In the top left and top right panels, two digraphs show the unit-weight votes of nine agents concerning two distinct jobs.
        Concerning the first job, the agents vote for the same agents as earlier:
        Agent 1 votes for agents 2, 6, and 9.
        Agent 2 for agents 3 and 5.
        Agent 3 for agent 8.
        Agent 4 for agent 3.
        Agent 5 for agent 1.
        Agent 6 for agents 5 and 7.
        Agent 7 for agents 2, 3, and 8.
        Agent 8 casts no votes.
        Agent 9 votes for agents 4, 6, and 8.
        Concerning the second job, the agents vote more sparsely:
        Agent 1 votes for agent 3.
        Agent 2 for agents 6 and 7.
        Agent 3 for agents 4 and 7.
        Agents 4, 5, and 6 cast no votes.
        Agent 7 votes for agents 6 and 8.
        Agent 8 for agent 2.
        Agent 9 for agent 6.
        In two rows below the digraphs, a partition system is shown twice, once for each job.
        Each row displays only the votes pertinent to the corresponding job.
        Again, the agents that appear on the right sides of the partitions are {1, 2, 3}, {4, 5, 6}, {7, 8, 9}, {1, 4, 7}, {2, 5, 8}, and {3, 6, 9}.
        In the partition system corresponding to the first job, the winners are agents 3, 5, 9, 7, 8, and 9.
        Agent 7 is marked gray to denote that it is unassigned.
        In the partition system corresponding to the second job, the winners are agents 2, 6, 7, 4, 5, and 3.
        Here, agents 3 and 5 are marked gray.
    }
\end{figure}

In this section, whenever $n,~m,~k$, and $\bfA\in \calA^m_n$ are fixed, we use $((S^1_1,S^1_2),\ldots,(S^k_1,S^k_2))$, $l(j)$, $r(j)$, $\hat{\score}_{S^p_1}(j;A_\ell),~\bfx^p$, and $\bfX$ for each $p\in [k]$, $\ell\in [m]$, and $j\in [n]$ to refer to the objects defined in $\alga_k$ for input $\bfA$.
We specify the input tuple of weight matrices as an argument when not clear from the context.
The following lemma, which plays an analogous role to \Cref{lem:bound-nice-n-k}, constitutes the main technical ingredient for the proof of \Cref{thm:assignment}.

\begin{lemma}\label{lem:bound-nice-n-k-assignment}
    Let $n,m,k\in \N$ with $k<n$ be such that $b\defas 2n/k \in \N$ and $m\leq b\leq k/2\in \N$.
    Then, $\alga_{k}$ is an impartial and $1/(2b)$-optimal $(n,m,k)$-assignment mechanism.
\end{lemma}

\begin{proof}
    We consider $n$ and $k$ as in the statement.
    We first note that $\alga_{k}$ is well-defined as we have $|\{p \in [k]: j\in S^{p}_2\}| = 2$ for every $j\in [n]$ and $b=2n/k\geq m$.
    On the other hand, since $x^p_\ell$ is a single agent for every $p\in [k]$ and $\ell\in[m]$, and $X_\ell\subseteq \bigcup_{p\in [k]}\{x^p_\ell\}$ for each $\ell\in [m]$, we have that $|X_\ell|\leq k$ for every $\ell\in [m]$.
    Further, due to the last step we have that for every $\bfA\in \calA^m_n$ and every $j\in [n]$ it holds that $|\{\ell\in [m]: j\in {(\alga_k(\bfA))}_\ell\}|\leq 1$, thus $\alga_{k}$ returns a feasible assignment.

    To see that $\alga_{k}$ is impartial, let $\bfA, \bfA'\in \calA^m_n$ and $j\in [n]$ be such that ${(A_\ell)}_{-j}={(A'_\ell)}_{-j}$ holds for every $\ell\in [m]$.
    Suppose $j\in {(\alga_{k}(\bfA))}_{\hat{\ell}}$ for some $\hat{\ell}\in [m]$.
    From the definition of the mechanism, we have that there is a $p\in [k]$ such that $j = x^p_{\hat{\ell}}(\bfA)$ and such that, if $j=x^q_{\ell}(\bfA)$ for $\ell\in [m]\setminus \{\hat{\ell}\}$ and $q\in \{l(j),r(j)\}\setminus \{p\}$, then $(\score(j;A_{\hat{\ell}}),\hat{\ell}) > (\score(j;A_{\ell}),\ell)$.
    Since $j\in S^p_2$ and ${(A_{\ell})}_{-j}={(A'_{\ell})}_{-j}$ for every $\ell\in [m]$, we have that $\hat{\score}_{S^p_1}(i;A_\ell) = \hat{\score}_{S^p_1}(i;A'_\ell)$ for every $p\in \{l(j),r(j)\}$, every $\ell\in[m]$, and every $i\in [n]$.
    Therefore, since the partial assignment $\bfx^p(\bfA)$ is defined as
    \[
        \bfx^p(\bfA) = \argmax_{\substack{\bfv\in {(S^p_2)}^m: v_\ell \not= v_{\ell'} \forall \ell\not=\ell'}} \Bigg(\sum_{\ell\in[m]} \hat{\score}_{S^p_1}(v_\ell;A_\ell), v_1, \ldots, v_m \Bigg),
    \]
    we obtain that $\bfx^p(\bfA)=\bfx^p(\bfA')$  for $p\in \{l(j),r(j)\}$.
    In particular, this yields $j = x^p_{\hat{\ell}}(\bfA')$ and that, if $j=x^q_{\ell}(\bfA')$ for $\ell\in [m]\setminus \{\hat{\ell}\}$ and $q\in \{l(j),r(j)\}\setminus \{p\}$, then $(\score(j;A'_{\hat{\ell}}),\hat{\ell}) > (\score(j;A'_{\ell}),\ell)$.
    Thus, $j\in {(\alga_{k}(\bfA'))}_{\hat{\ell}}$.
    Exchanging the roles of $\bfA$ and $\bfA'$ in the previous argument, we obtain that if $j\in {(\alga_{k}(\bfA'))}_{\hat{\ell}}$, then also $j\in {(\alga_{k}(\bfA))}_{\hat{\ell}}$.
    This yields ${(\alga_{k}(\bfA))}_{\hat{\ell}}\cap \{j\} = {(\alga_{k}(\bfA'))}_{\hat{\ell}} \cap \{j\}$.
    Since this reasoning is valid for all $\hat{\ell}\in [m]$, we conclude that ${(\alga_{k}(\bfA))}_\ell\cap \{j\} = {(\alga_{k}(\bfA'))}_\ell \cap \{j\}$ holds for every $\ell\in [m]$.

    For the remainder of this proof, we let $\bfA\in \calA_n$ be an arbitrary tuple of weight matrices.
    We start by observing that, for every $j\in [n]$ and $\ell\in [m]$,
    \begin{equation}
        \hat{\score}_{S^{l(j)}_1}(j; A_\ell) + \hat{\score}_{S^{r(j)}_1}(j; A_\ell) = \score_{S^{l(j)}_1}(j; A_\ell) + \score_{S^{r(j)}_1 \setminus S^{l(j)}_1}(j;A_\ell) = \score(j;A_\ell),\label{eq:sum-modified-indegrees-assignment}
    \end{equation}
    since for every $j\in [n]$, property~\ref{lem:partitions-2} of \Cref{lem:partitions} implies $S^{l(j)}_1\cup S^{r(j)}_1 = [n]\setminus \{j\}$.
    Furthermore, the definition of $\bfx^p$ yields
    \begin{equation}\label{eq:selected-vertex-assignment}
        \sum_{\ell\in [m]}\hat{\score}_{S^p_1}(x^p_\ell;A_\ell) \geq \hat{\score}_{S^p_1}(j;A_{\hat{\ell}})
    \end{equation}
    for every \(p\in [k]\), \(\hat{\ell}\in[m]\), and \(j\in S^p_2\).
    To see this, assume to the contrary that~\eqref{eq:selected-vertex-assignment} were not true for some $p\in[k],~ \hat{\ell}\in[m]$, and $j\in S^p_2$.
    Then, taking an alternative partial assignment $\bfz^p\in {(S^p_2)}^m$ defined as $z^p_{\hat{\ell}}=j$ and, for $\ell\in [m]\setminus\{\hat{\ell}\}$, $z^p_\ell$ equal to an arbitrary number in $[n]$ such that $z^p_{\ell}\not=z^p_{\ell'}$ for every $\ell,\ell'\in [m]$ with $\ell\not=\ell'$, we would obtain
    \[
        \sum_{\ell\in[m]}\hat{\score}_{S^p_1}(z^p_\ell;A_\ell) \geq \hat{\score}_{S^p_1}(j;A_{\hat{\ell}}) > \sum_{\ell\in [m]}\hat{\score}_{S^p_1}(x^p_\ell;A_\ell).
    \]
    However, this contradicts the definition of $\bfx^p$.
    Given these two facts, we claim that
    \[
        \sum_{\ell\in[m]}\big(\hat{\score}_{S^{l(j)}_1}(x^{l(j)}_\ell; A_\ell) + \hat{\score}_{S^{r(j)}_1}(x^{r(j)}_\ell; A_\ell)\big) \geq \max_{\ell\in [m]}\score(j;A_\ell) \label{eq:indegree-two-sets-assignment}
    \]
    for every \(j \in [n]\).
    To prove this, we fix $j\in [n]$ and observe that, for each $p\in \{l(j),r(j)\}$,
    inequality~\eqref{eq:selected-vertex-assignment} directly implies
    \[
        \sum_{\ell\in[m]} \hat{\score}_{S^p_1}(x^p_\ell; A_\ell) \geq \max_{\ell\in[m]} \hat{\score}_{S^p_1}(j;A_\ell).
    \]
    Therefore,
    \begin{align*}
        \sum_{\ell\in[m]}\big(\hat{\score}_{S^{l(j)}_1}(x^{l(j)}_\ell; A_\ell) + \hat{\score}_{S^{r(j)}_1}(x^{r(j)}_\ell; A_\ell)\big)
        & \geq \max_{\ell\in[m]} \hat{\score}_{S^{l(j)}_1}(j;A_\ell) + \max_{\ell\in[m]} \hat{\score}_{S^{r(j)}_1}(j;A_\ell) \\
        & \geq \max_{\ell\in [m]}\score(j;A_\ell),
    \end{align*}
    where the last inequality follows from equality~\eqref{eq:sum-modified-indegrees-assignment}.
    This concludes the proof of inequality~\eqref{eq:indegree-two-sets-assignment}.

    We next obtain a second inequality needed to conclude the approximation guarantee.
    Denoting by $\selected \defas \bigcup_{\ell\in [m]}{(\alga_k(\bfA))}_\ell$ the set of selected agents, for each $j\in \selected$ we consider the set
    \[
        L(j) \defas \left\{ \ell\in [m] ~\middle|~ j=x^{l(r)}_\ell \text{ or } j=x^{r(j)}_\ell \right\}
    \]
    containing the jobs $\ell$
    to which $j$ has been assigned to before the last $\mathbf{for}$ loop.
    Note that $|L(j)|\in \{1,2\}$ for every $j\in \selected$.
    We now observe that
    \[
        \score(\alga_{k}(\bfA)) = \sum_{j\in \selected} \max_{\ell\in L(j)} \score(j;A_\ell)
        \geq \frac{1}{2} \sum_{j\in \selected} \, \sum_{\ell\in L(j)} \score(j;A_\ell)
        = \frac{1}{2} \sum_{\ell\in [m]}\sum_{p\in [k]} \hat{\score}_{S^p_1}(x^p_\ell;A_\ell). \label{eq:indegree-alg-assignment}
    \]
    Indeed, the first equality follows from the last $\mathbf{for}$ loop in the definition of $\alga_k$, the inequality from the fact that a maximum of two values is at least their average, and the last equality from equality~\eqref{eq:sum-modified-indegrees-assignment} and the definition of $L(j)$.

    We now use inequalities\ \eqref{eq:indegree-two-sets-assignment}~and~\eqref{eq:indegree-alg-assignment} to conclude the bound stated in the lemma.
    For each $j\in \bigcup_{\ell\in[m]}{(\opt_k(\bfA))}_\ell$, we define $\ell(j) \defas \hat{\ell}$ for $\hat{\ell}\in [m]$ such that \(j\in {(\opt_k(\bfA))}_{\hat{\ell}}\).
    We further let $b\defas 2n/k$.
    Since $\big| \bigcup_{\ell\in[m]}{(\opt_k(\bfA))}_\ell \big| = km\leq n$, we know from property~\ref{lem:partitions-4} of \Cref{lem:partitions}, that there is a partition \({\dot\bigcup}_{\indx \in [b]} U_{\indx} = \bigcup_{\ell\in[m]}{(\opt_k(\bfA))}_\ell\) such that \(i \in S_2^p\) implies \(j \not\in S_2^p\) for all \(\indx \in [b]\), \(i, j \in U_{\indx}\) with $i\not=j$, and \(p \in [k]\).
    We obtain that, for every \(\indx \in [b]\),
    \begin{align}
        \score(\alga_{k}(\bfA)) & \geq \frac{1}{2} \sum_{\ell\in [m]}\sum_{p\in [k]} \hat{\score}_{S^p_1}(x^p_\ell;A_\ell) \nonumber \\
        & \geq \frac{1}{2} \sum_{\ell\in [m]} \sum_{j\in U_{\indx}} \left( \hat{\score}_{S^{l(j)}_1}(x^{l(j)}_\ell;A_\ell) + \hat{\score}_{S^{r(j)}_1}(x^{r(j)}_\ell;A_\ell) \right) \nonumber \\
        & \geq \frac{1}{2}\sum_{j\in U_{\indx}} \score(j;A_{\ell(j)}), \label{eq:indegree-clique-assignment}
    \end{align}
    where the first inequality follows from inequality~\eqref{eq:indegree-alg-assignment}, the second one from the fact that $\{l(i),r(i)\}\cap \{l(j),r(j)\}=\emptyset$ for every $\ell\in [b]$ and every $i,j\in U_{\ell}$ with $i\not= j$, and the last one from inequality~\eqref{eq:indegree-two-sets-assignment}.
    This yields
    \begin{align*}
        \score(\alga_{k}(\bfA)) & \geq \frac{1}{2} \max_{\indx\in [b]} \sum_{j\in U_{\indx}} \score(j;A_{\ell(j)}) \\
        & \geq \frac{1}{2b} \sum_{\indx\in [b]} \sum_{j\in U_{\indx}} \score(j;A_{\ell(j)})\\
        & = \frac{1}{2b} \sum_{\ell\in[m]} \score\left({(\opt_k(\bfA))}_\ell; A_\ell\right)
        = \frac{1}{2b}\score(\opt_k(\bf{A})),
    \end{align*}
    where the first inequality follows from~\eqref{eq:indegree-clique-assignment}, the second inequality from the observation that the maximum of a set of values is at least their average, the first equality from the fact that \( {\{U_\indx\}}_{\indx\in [b]} \) is a partition of $\bigcup_{\ell\in [m]}{(\opt_k(\bfA))}_\ell$ together with the definition of $\ell(j)$ for each $j$ in this set, and the last equality from the definition of $\score(\opt_k(\bf{A}))$. We conclude that \(\alga_{k}\) is $\alpha$-optimal for
    \[
        \frac{\score(\alga_{k}(\bfA))}{\score(\opt_k(\bfA))} \geq \frac{1}{2b}=\alpha. \qedhere
    \]
\end{proof}

The next lemma, analogous to \Cref{lem:bounds-nice-not-nice-n-k}, allows us to extend the bound given by \Cref{lem:bound-nice-n-k-assignment} to the case when $n$ and $k$ do not satisfy the conditions of \Cref{lem:bound-nice-n-k-assignment}.
Given $\tilde{n}, m, \tilde{k}\in \N$ with $m\tilde{k} < \tilde {n}$, an $(\tilde{n},m,\tilde{k})$-assignment mechanism $\alg$, and $n,k\in \N$ with $k\geq \tilde{k}$, $n\leq \tilde{n}$, $k<n$, and $mk\leq 2n$, this is achieved by the $(n,m,k)$-assignment mechanism that we formally describe as \Cref{alg:extension-alg-k-assignment} and whose output we denote as $\gena_{\alg,k}(\bfA)$ for an input tuple of weight matrices $\bfA\in \calA^m_n$.
\begin{algorithm}[t]
    \SetAlgoNoLine{}
    \KwIn{weight matrices tuple $\bfA\in \calA^m_n$}
    \KwOut{assignment $\bfX\in \calX_k$}
    define $\tilde{\bfA}\in \calA^m_{\tilde{n}}$ such that for each $\ell\in [m]$ we have
    \[
        {(\tilde{A}_\ell)}_{ij} = \begin{cases}
            {(A_\ell)}_{ij} & \text{ if } i,j\in [n],\\
            0 & \text{ otherwise.}
        \end{cases}
    \]
    {\bf return} $\alg(\tilde{\bfA})$
    \caption{$\gena_{\alg,k}(\bfA)$}%
    \label{alg:extension-alg-k-assignment}
\end{algorithm}
This algorithm simply extends $A_\ell$, for each $\ell\in [m]$, to the $\tilde{n}\times\tilde{n}$ matrix $\tilde{A}_\ell$ by adding $\tilde{n}-n$ many all-zero rows and columns, and then applies $\alg$ on $\tilde{A}$.
As before, whenever $\tilde{n},m,~n,~k,~\alg$, and $\bfA\in \calA^m_n$ are fixed, we use $\tilde{\bfA}$ to refer to the object defined in \Cref{alg:extension-alg-k-assignment}
for this input.
In a slight overload of notation, when we consider $\bfA'\in \calA_n$ as input, we write simply $\tilde{\bfA}'$.

\begin{lemma}\label{lem:bounds-nice-not-nice-n-k-assignment}
    Let $\tilde{n}, m, \tilde{k}\in\N$ with $m\tilde{k}\leq 2\tilde{n}$ be such that there exists an impartial and $\tilde{\alpha}$-optimal $(\tilde{n},m,\tilde{k})$-assignment mechanism $\alg$.
    Then, for every $k,n\in \N$ with $k\geq \tilde{k}$, $n\leq \tilde{n}$, $k<n$, and $mk\leq 2n$, $\gena_{\alg,k}$ is an impartial and $\alpha$-optimal $(n,m,k)$-assignment mechanism with $\alpha = (\tilde{k}/k) \tilde{\alpha}$.
\end{lemma}

\begin{proof}
    Let $n$, $m$, $k$, $\tilde{n}$ and $\tilde{k}$ be as in the statement.
    Let also $\alg$ denote the impartial and $\tilde{\alpha}$-optimal $(\tilde{n},m,\tilde{k})$-assignment mechanism.
    In order to see that $\gena_{\alg,k}$ is impartial, let $\bfA,\bfA'\in \calA^m_n$ and let $i\in [n]$ be such that ${(A_\ell)}_{-i}={(A'_\ell)}_{-i}$ for every $\ell\in [m]$.
    This implies ${(\tilde{A}_\ell)}_{-i}={(\tilde{A}'_\ell)}_{-i}$, thus the impartiality of $\alg$ yields that for every $\ell\in [m]$,
    \[
        {(\gena_{\alg,k}(\bfA))}_\ell \cap \{i\}
        = {(\alg(\tilde{\bfA}))}_\ell \cap \{i\}
        = {(\alg(\tilde{\bfA}'))}_\ell \cap \{i\}
        = {(\gena_{\alg,k}(\bfA'))}_\ell \cap \{i\}.
    \]

    To prove the approximation guarantee, we let $\bfA\in \calA^m_n$ be an arbitrary tuple of weight matrices and observe that
    \begin{equation}\label{eq:alpha-tilde1-assignment}
        \frac{\sigma(\gena_{\alg,k}(\bfA))}{\sigma(\opt_{\tilde{k}}(\tilde{\bfA}))} = \frac{\sigma(\alg(\tilde{\bfA}))}{\sigma(\opt_{\tilde{k}}(\tilde{\bfA}))} \geq \tilde{\alpha},
    \end{equation}
    where the equality follows from the definition of $\alga_{\alg,k}$ and the inequality follows from the $\tilde{\alpha}$-optimality of $\alg$.
    On the other hand, as $\tilde{k}\leq k$ and $\sigma(j,\tilde{A}_\ell)=0$ for every $j\not \in [n]$ and every $\ell\in [m]$, we know that
    \[
        \frac{\sigma(\opt_{k}(\bfA))}{k}
        = \frac{1}{k} \max_{\bfX \in \calX_k} \sum_{\ell\in[m]} \score(X_\ell;A_\ell)
        \leq \frac{1}{\tilde{k}} \max_{\bfX \in \calX_{\tilde{k}}} \sum_{\ell\in[m]} \score(X_\ell;A_\ell)
        = \frac{\sigma(\opt_{\tilde{k}}(\tilde{\bfA}))}{\tilde{k}},
    \]
    \ie{} the average score of the $k$ assigned agents of input $A$, under the best assignment, can be no larger than the average score of the $\tilde{k}$ assigned agents of input $\tilde{A}$, under the best assignment.
    To see this, note that we can obtain an assignment where we restrict to at most $\tilde{k}$ agents for each job from $\opt_k$ by deleting the $k-\tilde{k}$ agents per set ${(\opt_k(\bfA))}_\ell$ for each $\ell\in[m]$ that have the lowest score for this job.
    This can only increase the average score of the assignment.
    Plugging this inequality into~\eqref{eq:alpha-tilde1-assignment} concludes the proof as
    \[
        \frac{\sigma(\gena_{\alg,k}(\bfA))}{\sigma(\opt_k(\bfA))} \geq \frac{\tilde{k}}{k} \frac{\sigma(\gena_{\alg,k}(\bfA))}{\sigma(\opt_{\tilde{k}}(\tilde{\bfA}))} \geq \frac{\tilde{k}}{k} \tilde{\alpha}.\qedhere
    \]
\end{proof}

\Cref{thm:assignment} now follows from the last two lemmas.

\begin{proof}[Proof of \Cref{thm:assignment}]
    Let $n,m,k\in \N$ with $1<k<n,~ mk\leq 2n$ and $k-k\bmod 2\geq 2\sqrt{n}$.
    We define
    \[
        \tilde{k}\defas k-k\bmod 2,\quad \tilde{n} \defas \frac{k-k\bmod 2}{2} \left\lceil \frac{2n}{k-k\bmod 2} \right\rceil.
    \]
    It is clear that $\tilde{n},\tilde{k}$ are natural numbers with $\tilde{k}\leq k < n\leq \tilde{n}$, that $m\tilde{k}\leq 2\tilde{n}$, and that
    \[
        b\defas \frac{2\tilde{n}}{\tilde{k}} = \left\lceil \frac{2n}{k-k\bmod 2}\right\rceil\in \N.
    \]
    Moreover, we have that
    \[
        \tilde{n} = \frac{k-k\bmod 2}{2} \left\lceil \frac{2n}{k-k\bmod 2} \right\rceil \leq \frac{k-k\bmod 2}{2} \left\lceil \frac{2\frac{{(k-k\bmod2)}^2}{4}}{k-k\bmod 2} \right\rceil = \frac{\tilde{k}^2}{4},
    \]
    where the inequality follows from the condition $k-k\bmod 2\geq 2\sqrt{n}$ in the statement.
    This yields $b = 2\tilde{n}/\tilde{k} \leq \tilde{k}/2\in \N$.
    By \Cref{lem:bound-nice-n-k-assignment}, this implies that $\alga_{\tilde{k}}$ is an impartial and $\tilde{\alpha}$-optimal $(\tilde{n},m,\tilde{k})$-assignment mechanism with
    \[
        \tilde{\alpha} = \frac{1}{b} = \frac{1}{\left\lceil \frac{2n}{k-k\bmod 2}\right\rceil}.
    \]
    Since $\tilde{n},\tilde{k}\in \N$ are such that $\tilde{k}\leq k$ and $\tilde{n}\geq n$, \Cref{lem:bounds-nice-not-nice-n-k-assignment} implies that $\gena_{\alga_{\tilde{k}},k}$ is an impartial and $\alpha$-optimal $(n,m,k)$-assignment mechanism with
    \[
        \alpha = \frac{\tilde{k}}{k} \tilde{\alpha} = \frac{k-k\bmod 2}{k \left\lceil \frac{2n}{k-k\bmod 2}\right\rceil}. \qedhere
    \]
\end{proof}

\section{Discussion}

In an election among a group of peers, strategy-proofness is a desirable axiom that so far comes at a steep price.
If multiple winners are sought and the selection process must be deterministic, then the proposed impartial mechanism certainly signifies progress:
It allows the selection of more than two popular agents, can aggregate arbitrary ratings beyond the binary choice for or against approval, and improves the performance guarantee in the classical setting of unweighted votes from \(1/k\) to~\(k/(2n)\) for well-behaved combinations of \(k\)~and~\(n\), and to a value in that order in the wider range of \(k \geq 2 \sqrt{n}\).
Yet, this guarantee remains at a level that could be difficult to advertise to participants of the poll:
In the best case, when up to two thirds of the agents may be declared winners, the subset selected by the mechanism is only guaranteed to be supported by one third of the score that the highest-rated subset receives.
We showed that this minimum performance is indeed attained in adverse instances, and already for votes carrying unit weight.
This renders the analysis of the presented mechanism tight in the setting of approval voting.
On the other hand, the best known upper bound on the performance ratio achievable by any deterministic mechanism, \((k - 1)/k\), remains unabatedly optimistic and suggests the possibility of mechanisms that are truly practical.
Narrowing the gap in the approximate optimality thus remains a central challenge in the study of impartial mechanisms, both for deterministic and randomized selection.

We analyzed our mechanism in a rather general setting:
Agents can signal an arbitrary level of approval for any number of peers, possibly none of them.
While the adverse instance presented in \Cref{sec:selection} weighs all votes equally, it is remarkably sparse, with just three out of 72 potential votes being cast.
Two interesting variants of the impartial selection problem studied in the literature are the \emph{no-abstentions} and the \emph{single-nomination} settings, both of which exclude this particular instance.
In the former, each agent is required to cast at least one vote, while in the latter, it must be exactly one.
A slightly more general rule could pose lower and upper bounds on the total nomination weight that must be distributed by each agent; a quite reasonable restriction when weighted votes are allowed.
It is conceivable that our mechanism fares much better in these restricted settings and attains a performance ratio more easily communicated to practitioners.
Likewise, it is possible that the mechanism fares well in expectation when the agents' indices are permuted at random, thus turning it into an inexact and randomized procedure.
Arguably, the property that it processes every vote exactly once may make it more attractive in practice than existing randomized mechanisms that fill the selection budget exactly but omit some votes.
The theoretical analysis of the presented mechanism under these scenarios is another promising direction for further studies.

\section*{Acknowledgements}

Javier Cembrano was supported by the Deutsche Forschungsgemeinschaft (DFG, German Research Foundation) under project number 431465007.
Svenja M.\ Griesbach and Maximilian J.\ Stahlberg were supported by the Deutsche Forschungsgemeinschaft under Germany's Excellence Strategy -- The Berlin Mathematics Research Center MATH+ (EXC-2046/1, project~390685689). 

The authors thank the participants of the Second Workshop on Trends in Algorithmic Discrete Mathematics for fruitful discussions.
In particular, we thank Max Klimm for proposing the problem.
We also thank the anonymous reviewers for suggesting improvements to the article.

\renewcommand{\bibfont}{\setstretch{0.85}}

\bibliographystyle{abbrvnat}
\bibliography{manuscript}

\end{document}